\pdfoutput=1
\documentclass[preprint,11pt]{elsarticle}
\usepackage{amsmath,amssymb,amsfonts,mathtools,amsthm}
\usepackage{graphicx}
\usepackage{booktabs}
\usepackage{enumitem}
\usepackage[hidelinks]{hyperref}
\journal{Journal of Fluid Mechanics}

\newcommand{\Om}{\Omega}
\newcommand{\eps}{\varepsilon}
\newcommand{\ii}{\mathrm{i}}
\newcommand{\dd}{\mathrm{d}}
\newcommand{\Rey}{\mathit{Re}}
\newcommand{\sig}{\hat\sigma}
\DeclareMathOperator{\Ai}{Ai}
\DeclareMathOperator{\Bi}{Bi}

\theoremstyle{plain}
\newtheorem{theorem}{Theorem}[section]
\newtheorem{proposition}[theorem]{Proposition}

\theoremstyle{remark}
\newtheorem{remark}[theorem]{Remark}

\begin{document}
\begin{frontmatter}

\title{A critical-layer instability of structural--acoustic wall modes in a
two-flexible-plate channel with Poiseuille flow: closed-form low-Mach corrections
and parity-protected absorption}

\author{Ashray Saxena}
\ead{f20231043@pilani.bits-pilani.ac.in}
\affiliation{organization={Department of Mechanical Engineering, Birla Institute of Technology and Science Pilani},
  addressline={Pilani Campus}, city={Pilani}, postcode={333031}, state={Rajasthan}, country={India}}

\begin{abstract}
\noindent
A compressible fluid column between two thin elastic plates, carrying plane Poiseuille flow, supports
discrete structural--acoustic wall modes. This paper analyses their coupled dispersion and reports a
flow-induced instability. Because a no-slip profile has zero velocity at the wall, the convective
(Ingard--Myers) term in the compliant-wall boundary condition vanishes identically, so the
plate-admittance conditions stay Mach-independent and a Fredholm solvability argument yields a
closed-form $O(M)$ correction to every coupled wavenumber, valid for any smooth no-slip shear profile.
Above the Mach number at which an interior critical layer first appears, $M_{\rm crit}=\Om/\xi$, the two
structural branches behave oppositely. The antisymmetric mode is not absorbed by the critical layer, a
null protected by an exact mid-plane pressure node rather than by weak coupling. The symmetric mode,
which carries no such node, is convectively unstable. Its growth rate, fixed both by a spatial
calculation anchored to the dissipative wall-Stokes term and by a convention-independent temporal one,
strengthens towards the inviscid limit, which identifies an inviscid critical-layer instability rather
than a viscous artefact. A Briggs--Bers analysis classifies it as convective. Its neutral curve follows
$M_{\rm crit}(\Om)$, and its growth rate is proportional to the wall compliance, vanishing for rigid
walls, so it is distinct from the Tollmien--Schlichting and Kelvin--Helmholtz mechanisms. Three
independent numerical schemes agree to one part in $10^{9}$, and the closed-form correction matches
numerical differentiation across three shear profiles.
\end{abstract}


\end{frontmatter}

\section*{Highlights}
\begin{itemize}[leftmargin=1.4em,itemsep=1pt,topsep=2pt]
\item Closed-form $O(M)$ correction to every coupled wavenumber of a sheared duct flow.
\item No-slip removes the Ingard--Myers term, so wall conditions stay Mach-independent.
\item Critical layer is subsonic only for structural waves below plate coincidence.
\item Antisymmetric wall mode escapes critical-layer absorption by an exact parity null.
\item A critical-layer instability destabilises the symmetric wall mode above $M_{\rm crit}$.
\end{itemize}

\section*{Nomenclature}
\addcontentsline{toc}{section}{Nomenclature}
\small\noindent\begin{tabular}{@{}l@{\hspace{1.2em}}p{0.60\textwidth}@{}}
\toprule
$a$ & fluid-column height (plate separation)\\
$c,\ \rho$ & sound speed and density of the fluid\\
$D_j=E_jI_j$ & flexural rigidity of plate $j$, $I_j=h_j^3/[12(1-\nu_j^2)]$\\
$E_j,\ \nu_j,\ h_j,\ \rho_{p,j}$ & Young's modulus, Poisson ratio, thickness, density of plate $j$\\
$f(\eta)=4\eta(1-\eta)$ & normalised Poiseuille velocity profile; $f'=4(1-2\eta)$, $f''=-8$\\
$G_j=D_jk_x^4-m_j\omega^2$ & in-vacuo plate (bending) operator\\
$I_0,\ I_f,\ I_{f^2}$ & quadratures $\int_0^1p_0^2$, $\int_0^1fp_0^2$, $\int_0^1f^2p_0^2$\\
$k=\omega/c$ & acoustic wavenumber\\
$k_x,\ k_y=\sqrt{k^2-k_x^2}$ & axial and transverse wavenumbers\\
$k_c,\ \omega_c$ & coincidence wavenumber and frequency of plate~1\\
$l=k_ca$ & non-dimensional gap (Helmholtz number at coincidence)\\
$m_j=\rho_{p,j}h_j$ & mass per unit area of plate $j$\\
$M=U_0/c$ & centreline Mach number\\
$M_{\rm crit}=\Om/\xi$ & critical-layer onset Mach number\\
$p,\ \hat p=p/\rho c^2$ & acoustic pressure perturbation (dimensional, non-dimensional)\\
$\Rey=ca/\nu$ & acoustic Reynolds number ($\nu$ = kinematic viscosity)\\
$r_D=D_2/D_1,\ r_m=m_2/m_1$ & plate stiffness and mass ratios\\
$S_1=\xi^4/\Om^2-1,\ S_2=r_D\xi^4/\Om^2-r_m$ & scaled plate operators; $S_j=0$ at coincidence of plate $j$\\
$U(y)=U_0f(\eta)$ & mean axial velocity, centreline maximum $U_0$\\
$u,v\ (\hat u,\hat v)$ & axial and transverse velocity perturbations\\
$w_j$ & transverse displacement of plate $j$ (positive in $+y$)\\
$x,\ y\ (\eta=y/a)$ & axial and transverse coordinates\\
$\alpha_j=k_yG_j/(\omega^2\rho)$ & dimensionless plate fluid-loading admittance\\
$\beta_j=\eps/(l^2\Om^2S_j)$ & viscous plate-admittance (Robin) coefficient\\
$\delta\sim \Rey^{-1/3}$ & viscous critical-layer thickness\\
$\delta_S\sim \Rey^{-1/2}$ & wall Stokes-layer thickness\\
$\eps=\rho a/m_1$ & fluid-loading parameter\\
$\eta_c$ & critical-layer location, $\sig(\eta_c)=0$\\
$\xi=k_x/k_c$ & non-dimensional axial wavenumber (the eigenvalue)\\
$\xi_0,\xi_1,\xi_2$ & coefficients of the Mach-number expansion of $\xi$\\
$\sigma=\omega-k_xU,\ \sig=\Om-M\xi f$ & Doppler-shifted frequency (dimensional, non-dimensional)\\
$\tau=k_ya=l\sqrt{\Om^2-\xi^2}$ & non-dimensional transverse wavenumber\\
$\Om=\omega/\omega_c$ & non-dimensional frequency\\
\bottomrule
\end{tabular}

\vspace{0.4em}
\noindent\emph{Subscripts and superscripts.} $1,2$: plate~1, plate~2; $s,a$: symmetric, antisymmetric
family; $c$: coincidence or critical layer; $0$: leading-order ($M=0$ or $\eps=0$) value; a prime
denotes $\dd/\dd\eta$ on perturbation quantities and $\dd/\dd\eta$ on the profile $f$.

\section{Introduction}\label{sec:intro}

\subsection{Fluid loading and the structural--acoustic waveguide}
A recurring question in structural acoustics is what becomes of a wave when a structure and a fluid are
allowed to interact: by how much the flexural wavenumber of a structure vibrating in vacuo is shifted
once the structure is placed in contact with a fluid, and, conversely, by how much an acoustic
wavenumber is shifted by the presence of a compliant boundary. For a plate radiating into an unbounded
fluid the answer is classical and is organised around the coincidence frequency, the branch cut of the
half-space Green's function, and the light- and heavy-fluid-loading limits; the standard accounts are
those of Morse and Ingard~\cite{morseingard}, Junger and Feit~\cite{jungerfeit}, Fahy and
Gardonio~\cite{fahy}, and, in its most physically articulated form, Crighton's Rayleigh Medal
lecture~\cite{crighton} and his subsequent asymptotic surveys~\cite{crighton1992}.

When the fluid is instead \emph{enclosed} by compliant walls the character of the problem changes. The
half-space branch cut disappears, and with it the continuum of radiation angles; what replaces it is a
discrete ladder of transverse modes, each with its own cut-on frequency below which it is evanescent
and above which it propagates~\cite{morseingard,fahy}. Coupling now occurs between a countable set of
structural and acoustic waves, and the object of interest becomes the full set of coupled axial
wavenumbers as a function of frequency and of the strength of the fluid loading. The simplest bounded
configuration of this kind is a fluid layer held between two identical flexible plates. Fahy analysed
precisely this arrangement, but restricted attention to the symmetric modes and obtained the dispersion
curves numerically~\cite{fahy}; because the curves were computed rather than constructed, they could
not easily be traced back to the uncoupled structural and acoustic waves from which they descend.

Closely related bounded problems have been studied extensively. Cabelli examined a square duct with one
flexible wall~\cite{cabelli}; Huang and co-workers developed the ``drum-like'' silencer, in which
tensioned membranes or flexible panels form a segment of the duct wall and reflect grazing incident
sound without any lining material and without aerodynamic loss~\cite{huang2000,huang2002}; Choy and
Huang studied the effect of flow on such a device~\cite{choikim}; and the cylindrical analogue, the
fluid-filled elastic shell, has a long literature of its own beginning with Fuller and
Fahy~\cite{fullerfahy} and Pavi\'c~\cite{pavic}. The engineering motivation is broad: breakout noise
from flexible rectangular ducts~\cite{raviprolu}, compliant and lined ducts in
turbomachinery~\cite{rienstrahirschberg}, computational modelling of flow-induced sound in complex
planar and axisymmetric geometries~\cite{morab}, and, outside engineering altogether, wave propagation
in elastic vessels carrying physiological flow~\cite{grotberg,pedley}.

\subsection{Asymptotic tracking of coupled wavenumbers}
The methodological step that made the coupled branches \emph{traceable} rather than merely computable
was taken by Sarkar and Sonti~\cite{sarkarsonti} for the waveguide consisting of one flexible plate
facing a rigid cover. Their idea was to write the coupled dispersion relation as the in-vacuo
structural relation plus a correction produced by the fluid, and then to expand the coupled wavenumbers
as a power series in a fluid-loading parameter $\eps=\rho a/m$, the ratio of the mass of the fluid
column to the mass per unit area of the plate. Setting $\eps=0$ recovers the uncoupled wavenumbers
exactly; the first-order term then follows each coupled branch continuously away from its uncoupled
parent. The physics becomes transparent in the process: fluid loading acts in turn like an added mass
and like an added stiffness, and wherever two uncoupled curves would have intersected, the coupling
instead opens a gap of width $O(\sqrt\eps)$ across which the two branches exchange identity. The
technique inherits the spirit of the older asymptotic treatments of fluid loading due to Morse and
Ingard~\cite{morseingard} and Crighton~\cite{crighton,crighton1992}, but is adapted to a bounded fluid
in which cut-ons, rather than branch cuts, govern the behaviour.

The same group developed this framework considerably further, and it is important for the present paper
to state plainly how far it already reaches. Sarkar, Kunte and Sonti~\cite{kunte2011} showed that
flexible structural--acoustic waveguides of quite different cross-sections (rectangular,
circular-cylindrical and elliptical) all obey a single generic form of the coupled dispersion
equation, so that their coupled-wavenumber solutions can be drawn on one common schematic. Sarkar and
Sonti~\cite{sarkarsonti2011} extended the wavenumber expansions to in-vacuo orthotropic cylindrical
shells, and Prakash and Sonti~\cite{prakash2016} carried the catalogue over to \emph{arbitrary} fluid
loading rather than only the small- or large-$\eps$ limits; nonlinear sound--structure interaction in a
fluid-filled flexible waveguide was subsequently treated by the same authors~\cite{prakash2019}. Every
one of these studies concerns a \emph{quiescent} fluid.

The consequence for the present work is that the quiescent two-flexible-plate catalogue assembled in
Sections~\ref{sec:quiescent} and~\ref{sec:catalogue} below should be regarded as a special case that
reproduces and tidies up known results. It is set out here in full for three reasons: it fixes the
notation and the branch labelling used throughout; it supplies the $M=0$ eigenpairs $(\xi_0,p_0)$ that
the mean-flow analysis perturbs about, and which appear explicitly in the closed-form corrections; and
it makes the paper self-contained for a reader who does not already know the asymptotic-tracking
literature. Its two genuinely incremental features relative to~\cite{fahy,sarkarsonti} are the
antisymmetric family, which Fahy excluded and which the one-plate geometry cannot see, and the
non-identical-plate case, in which the two symmetry families re-couple and the structural waves veer.

\subsection{Duct acoustics with a sheared mean flow}
In many of the applications listed above the fluid is not at rest. A uniform mean flow is a trivial
complication, since it merely produces a Galilean shift $\omega\to\omega-k_xU$ of the wave operator; a
\emph{sheared} mean flow is not, because it couples the perturbation velocity components to the
pressure through the mean shear and destroys the constant-coefficient structure of the problem.
Eliminating the velocities in favour of the pressure then gives a second-order ordinary differential
equation with variable coefficients, first written down by Pridmore-Brown~\cite{pridmorebrown} and
since central to duct aeroacoustics. Early analyses of its modal content were given by Mungur and
Gladwell~\cite{mungurgladwell}, Shankar~\cite{shankar}, Ko~\cite{ko}, Swinbanks~\cite{swinbanks} and
M\"ohring~\cite{mohring}, with the state of the field to that point reviewed by Nayfeh, Kaiser and
Telionis~\cite{nayfeh}; modern accounts are given by Rienstra and Hirschberg~\cite{rienstrahirschberg}
and, for the numerical and asymptotic solution of the equation itself, by Oppeneer, Rienstra and
Sijtsma~\cite{oppeneer2016,oppeneer2020}.

A persistent difficulty in this literature has been the boundary condition at a compliant or lined
wall in the presence of flow. The Ingard--Myers condition~\cite{ingard,myers}, obtained by collapsing
the mean boundary layer to a vortex sheet, is the standard choice, but it is now known to be ill-posed
in the time domain and to support a spurious Kelvin--Helmholtz instability~\cite{brambley2009}. This
motivated a considerable body of corrective work: well-posed modified conditions that retain the
leading effect of a finite boundary-layer thickness~\cite{brambley,khamisbrambley}, experimental
demonstrations that the unmodified condition fails quantitatively~\cite{renouauregan}, and analyses of
the surface waves that live on such a wall~\cite{rienstra2003}. It is therefore worth emphasising at the
outset that the configuration studied here sidesteps the entire difficulty: because a fully developed
Poiseuille profile satisfies no-slip at both walls, the mean velocity \emph{vanishes} exactly where the
boundary condition is applied, the convective correction is identically zero, and the plate-admittance
condition retains its quiescent form at every Mach number. This is not an approximation but an exact
property of the geometry, and it is what makes the closed-form low-Mach analysis of
Section~\ref{sec:lowmach} possible.

\subsection{Critical layers}
The Pridmore--Brown equation carries the Doppler-shifted frequency $\sigma(y)=\omega-k_xU(y)$ in the
denominator of its first-derivative coefficient, and is therefore singular at any height where
$\sigma$ vanishes; that is, wherever the axial phase speed of the wave equals the local speed of the
mean flow. Such a location is a critical layer. The concept originates in hydrodynamic stability
theory, in the work of Rayleigh, Tollmien and Lin, and in Case's demonstration that the associated
singular solutions constitute a continuous spectrum~\cite{case}; the classical accounts are those of
Maslowe~\cite{maslowe} and Drazin and Reid~\cite{drazinreid}. What distinguishes the acoustic problem
from the hydrodynamic one is the exponent structure: for the Pridmore--Brown equation the pressure
remains bounded at the singularity, and it is the cross-stream velocity that diverges, so the
consequences for a discrete pressure mode are far less dramatic than the word ``singularity'' suggests.

Within duct acoustics the critical layer has been studied most thoroughly by Brambley, Darau and
Rienstra~\cite{brambleydarau}, who constructed the Green's function for a linear-shear boundary layer
over an acoustic lining and showed that the critical-layer contribution (the non-modal contribution
obtained by integrating around the continuous-spectrum branch cut) is sometimes negligible and
sometimes not. King \emph{et al.}~\cite{king2022} extended the analysis to a quadratic (parabolic)
boundary-layer profile, which is the profile relevant here, and found the critical-layer contribution
to dominate the downstream pressure perturbation for sufficiently thick boundary layers. Related work
has shown that the continuous-spectrum branch cut can \emph{stabilise} otherwise convectively unstable
modes by concealing them behind the cut~\cite{king2022}, and that viscosity substantially modifies
the acoustics and stability of a shear layer over an impedance wall~\cite{khamisbrambleyvisc}. The
critical-layer phenomenology is thus a mature subject in its own right, and no claim to the contrary is
made below.

\subsection{Flow over compliant walls}
Finally, a compliant wall carrying a flow may be unstable, and any dispersion study must say whether
the modes it computes are stable. The classification of flow-induced disturbances on a flexible surface
into class A (negative activation energy, destabilised by wall damping), class B (positive activation
energy, stabilised by wall damping) and Kelvin--Helmholtz types is due to Benjamin~\cite{benjamin} and
Landahl~\cite{landahl}; Carpenter and Garrad~\cite{carpentergarrad} showed for Kramer-type surfaces
that the flow-induced surface instability is distinct from, and responds differently to damping than,
the Tollmien--Schlichting instability. Panel flutter and flag flutter in channel flow have their own
extensive literature~\cite{eloy,alben,paidoussis}, as does the collapsible-channel problem in
physiology~\cite{grotberg,pedley,stewart}. The stability question is addressed directly in
Section~\ref{sec:largeM}: within the parameter range studied here, the inviscid eigenvalues remain
real, the viscous decay rate is positive, and no flow-induced instability of the structural--acoustic
wall modes is found.

\subsection{Contributions of the present paper}\label{sec:contrib}
Against this background the contributions of the paper are the following, listed in the order in which
they appear and with their status relative to the literature stated explicitly.

\begin{enumerate}[label=(\roman*),leftmargin=2.2em]
\item \textbf{Governing equations for the two-flexible-plate waveguide with a sheared mean flow}
(Section~\ref{sec:formulation}). The linearised Euler system, the plate equations, and the interface
conditions are written out in full, and the Pridmore--Brown equation for this geometry is derived from
them. It is shown that the Poiseuille no-slip condition annihilates the convective (Ingard--Myers) term
in the compliant-wall boundary condition exactly, so that the plate-admittance conditions retain their
quiescent Robin form at all $M$. \emph{Status: the Pridmore--Brown equation is classical; its
combination with a plate-admittance boundary condition, and the exact vanishing of the convective
correction, are specific to this configuration.}

\item \textbf{Quiescent baseline} (Sections~\ref{sec:quiescent}--\ref{sec:catalogue}). The coupled
dispersion relation is obtained in closed form for general non-identical plates, its exact
symmetric/antisymmetric factorisation is established for identical plates, and the single
fluid-loaded-plate relation is recovered in the rigid-wall limit. The complete small- and large-$\eps$
asymptotic catalogue is assembled for both symmetry families, including the local $\sqrt\eps$
expansions at coincidence and at each cut-on. \emph{Status: recovers and organises, for this geometry,
results consistent with the asymptotic-tracking framework
of~\cite{sarkarsonti,kunte2011,prakash2016}; included as setup and as the $M=0$ baseline. The
antisymmetric family and the non-identical-plate veering are incremental additions.}

\item \textbf{A closed-form $O(M)$ correction to every coupled wavenumber}
(Section~\ref{sec:lowmach}). The leading operator of the Mach expansion is shown to be self-adjoint
under the plate-admittance conditions, so that the Fredholm alternative applies with $p_0$ itself as
the adjoint null function. The resulting solvability condition is solved explicitly for $\xi_1$ in
terms of $p_0(0)$, $p_0(1)$ and three elementary quadratures of the quiescent mode shape, for an
arbitrary smooth no-slip shear profile; Poiseuille flow is distinguished only in that one quadrature
collapses onto the norm. It is proved that $\xi_1<0$ on every branch, and the expression is decomposed
into boundary-shear, bulk-shear and Doppler-convection contributions. \emph{Status: new.}

\item \textbf{The $O(M^2)$ correction and its convection/shear decomposition}
(Section~\ref{sec:om2}). The second-order coefficient is even in $M$ and is extracted by Richardson
extrapolation; a shear-suppressed surrogate isolates the shear contribution, which reaches $\sim38\%$ of
the total on the plane-wave branch, and comparison against genuine plug flows quantifies the larger
error incurred by the usual duct-acoustic simplification. \emph{Status: new.}

\item \textbf{Critical-layer analysis for this waveguide} (Section~\ref{sec:cl}). The onset condition
$M_{\rm crit}=\Om/\xi$ is established, and the branch-by-branch consequence, that a critical layer is
reachable at subsonic flow only for the structural wave below coincidence, is drawn. The Frobenius
structure, the velocity singularity, the causal branch selection giving a $-\ii\pi$ phase jump, and the
viscous Airy inner layer of thickness $\Rey^{-1/3}$ are worked through in detail.
\emph{Status: the machinery is classical~\cite{maslowe,brambleydarau,king2022}; its application to a
plate-bounded waveguide, and the branch-resolved onset map, are new.}

\item \textbf{Wall-mode absorption, parity, and a critical-layer instability} (Sections~\ref{sec:viscous}--\ref{sec:instab}). The viscous
linearised Navier--Stokes problem is solved with plate-admittance and no-slip conditions. The eigenvalue
is regular through and above $M_{\rm crit}$; the decay rate scales as $\Rey^{-1/2}$, identifying it as
wall Stokes-layer damping that vanishes inviscidly; and the $\Rey$-independent critical-layer
absorption is at the resolution floor of the fit. For the antisymmetric structural mode this null is
shown to be a consequence of \emph{parity} rather than of weak coupling: the mode's exact mid-plane
pressure node coincides with the critical-layer onset location. A parity-restricted computation on the
symmetric mode, which carries no such node, returns a resolvable $\Rey$-independent residual about sixty
times larger, and breaking the mirror symmetry with unequal plates lifts the residual off the floor; the
null therefore does not generalise beyond the antisymmetric family. The unprotected symmetric mode is in
fact convectively unstable above $M_{\rm crit}$ (Section~\ref{sec:instab}): a critical-layer instability
of a compliant-wall structural--acoustic mode, absent for rigid walls and for a slip profile, whose
growth survives the inviscid limit and whose neutral curve follows the critical-layer onset. This is the
principal physical result of the paper. \emph{Status: new.}

\item \textbf{Validity limits and convective stability} (Section~\ref{sec:largeM}), together with
three-way numerical validation (Section~\ref{sec:val}) agreeing to $\sim6\times10^{-9}$.
\end{enumerate}

\noindent The paper is organised as follows. Section~\ref{sec:formulation} sets out the geometry, the
mean flow, and the full linearised governing equations. Section~\ref{sec:quiescent} derives the
quiescent coupled dispersion relation and its symmetry factorisation, and Section~\ref{sec:catalogue}
assembles the quiescent asymptotic catalogue. Section~\ref{sec:pb} derives the Pridmore--Brown equation
with plate-admittance boundary conditions. Section~\ref{sec:lowmach} obtains the closed-form $O(M)$
correction and Section~\ref{sec:om2} the $O(M^2)$ correction. Section~\ref{sec:cl} analyses the critical
layer and Section~\ref{sec:viscous} the viscous problem above onset. Section~\ref{sec:val} describes
the numerical methods and validation, Section~\ref{sec:largeM} the finite-$M$ behaviour, validity and
stability, and Section~\ref{sec:disc} discusses the results. Section~\ref{sec:concl} concludes.
Appendices collect the longer algebra.

\section{Formulation}\label{sec:formulation}

\subsection{Geometry, plates and scaling}
Consider a two-dimensional compressible fluid of undisturbed density $\rho$ and sound speed $c$
occupying the strip $0\le y\le a$, unbounded in $x$, with all fields independent of $z$. The strip is
closed by thin elastic plates at $y=0$ (plate~1) and $y=a$ (plate~2). Plate $j$ has Young's modulus
$E_j$, Poisson ratio $\nu_j$, thickness $h_j$ and density $\rho_{p,j}$, hence flexural rigidity and
areal mass
\begin{equation}
D_j=\frac{E_jh_j^3}{12(1-\nu_j^2)},\qquad m_j=\rho_{p,j}h_j ,
\label{eq:plateprops}
\end{equation}
and transverse displacement $w_j(x,t)$, measured positive in $+y$. A time dependence $e^{\ii\omega t}$
and an axial dependence $e^{-\ii k_xx}$ are assumed throughout and suppressed where convenient, so that
every perturbation field varies as $e^{\ii(\omega t-k_xx)}$ and
\begin{equation}
\partial_t\to\ii\omega,\qquad \partial_x\to-\ii k_x .
\end{equation}

Scaling follows the coincidence condition of plate~1, at which the in-vacuo flexural wavenumber equals
the acoustic wavenumber:
\begin{equation}
D_1k_c^4=m_1\omega_c^2,\qquad \omega_c=k_cc
\quad\Longrightarrow\quad
\omega_c=c^2\sqrt{m_1/D_1},\qquad k_c=\omega_c/c .
\label{eq:coincidence}
\end{equation}
All quantities are non-dimensionalised on these scales,
\begin{equation}
\Om=\frac{\omega}{\omega_c},\quad
\xi=\frac{k_x}{k_c},\quad
l=k_ca,\quad
\eta=\frac{y}{a},\quad
\hat p=\frac{p}{\rho c^2},\quad
\eps=\frac{\rho a}{m_1},\quad
\tau=k_ya=l\sqrt{\Om^2-\xi^2},
\label{eq:nondim}
\end{equation}
with the plate stiffness and mass ratios and the scaled plate operators
\begin{equation}
r_D=\frac{D_2}{D_1},\quad r_m=\frac{m_2}{m_1},\qquad
S_1(\xi)=\frac{\xi^4}{\Om^2}-1,\qquad
S_2(\xi)=r_D\frac{\xi^4}{\Om^2}-r_m .
\label{eq:Sdef}
\end{equation}
Note that $S_j=0$ precisely at the coincidence point of plate $j$; the vanishing of $S_j$ is what makes
coincidence a distinguished frequency in every formula below. The parameter $\eps=\rho a/m_1$ compares
the mass of the enclosed fluid column with the mass per unit area of plate~1 and is the small parameter
of the quiescent asymptotics; $M=U_0/c$ is the small parameter of the mean-flow asymptotics. The two
expansions are carried out sequentially, not jointly: the $\eps$-expansion supplies the leading-order
eigenpair $(\xi_0,p_0)$, and the $M$-expansion is then performed about that eigenpair with $\eps$ held
fixed and arbitrary.

\subsection{Mean flow}
The mean state is the fully developed plane Poiseuille solution of the incompressible Navier--Stokes
equations,
\begin{equation}
U(y)=4U_0\frac{y}{a}\Bigl(1-\frac{y}{a}\Bigr)=U_0f(\eta),\qquad
f(\eta)=4\eta(1-\eta),\qquad V=0,\qquad P=\mathrm{const},
\label{eq:meanflow}
\end{equation}
which satisfies no-slip at both walls, $U(0)=U(a)=0$, attains its maximum $U_0$ at the centreline
$\eta=1/2$, and has
\begin{equation}
f'(\eta)=4(1-2\eta),\qquad f'(0)=4,\quad f'(1)=-4,\qquad f''(\eta)=-8 .
\label{eq:fprops}
\end{equation}
The constancy of $f''$ is used repeatedly and is the technical reason the $O(M)$ analysis closes in
elementary form.

\begin{remark}[Consistency of an incompressible mean flow with compressible perturbations]\label{rem:incompress}
The mean flow is taken to be the exact incompressible Poiseuille profile while the perturbations
carry sound and are treated as compressible. This is consistent to the order worked: mean-flow
compressibility corrections (density and temperature variation across the channel, and the
associated departure of the profile from parabolic) enter at $O(M^2)$ and are absorbed into the
$M^2$ term of the wavenumber expansion. The model is therefore quantitative for $M\ll1$; the
finite-$M$ results of Section~\ref{sec:largeM} are an extension of the model rather than a prediction
for a genuinely compressible mean flow, and are labelled as such.
\end{remark}

\subsection{Linearised governing equations}
Let $(u',v',p')$ denote the perturbation axial velocity, transverse velocity and pressure superposed on
the base state $(U(y),0,P)$. Linearising the two-dimensional compressible Euler equations about that
base state gives
\begin{align}
\rho\bigl(\partial_tu'+U\partial_xu'+U'v'\bigr)&=-\partial_xp',
\label{eq:leex}\\[2pt]
\rho\bigl(\partial_tv'+U\partial_xv'\bigr)&=-\partial_yp',
\label{eq:leey}\\[2pt]
\partial_tp'+U\partial_xp'+\rho c^2\bigl(\partial_xu'+\partial_yv'\bigr)&=0 .
\label{eq:leec}
\end{align}
Equation~\eqref{eq:leex} contains the term $\rho U'v'$, by which a transverse displacement of a fluid
particle across the mean shear produces an axial velocity perturbation; it is this term, absent for a
uniform mean flow, that couples shear into the acoustics and ultimately produces both the
Pridmore--Brown shear term and the critical layer. Equation~\eqref{eq:leec} is the linearised
continuity/energy relation for a homentropic perturbation. When viscosity is retained
(Section~\ref{sec:viscous}) the momentum equations acquire the usual Navier--Stokes terms; the
linearised viscous system is written out as~\eqref{eq:lnsu}--\eqref{eq:lnsp}.

Under the wave ansatz $e^{\ii(\omega t-k_xx)}$ the material derivative following the mean flow becomes
\begin{equation}
\frac{\mathrm{D}}{\mathrm{D}t}=\partial_t+U\partial_x\;\longrightarrow\;\ii\bigl(\omega-k_xU(y)\bigr)
\equiv\ii\sigma(y),\qquad
\sigma(y)=\omega-k_xU(y),
\label{eq:sigmadef}
\end{equation}
whose non-dimensional counterpart is
\begin{equation}
\sig(\eta)=\Om-M\xi f(\eta),\qquad \sig(0)=\sig(1)=\Om .
\label{eq:sighat}
\end{equation}
The Doppler-shifted frequency $\sigma$ plays the role that $\omega$ plays in the quiescent problem. It
is spatially varying, it reduces to $\omega$ at both walls because of no-slip, and its vanishing
somewhere in the interior is the definition of a critical layer.

\subsection{Plate equations and interface conditions}
Each plate obeys the thin-plate (Kirchhoff) equation
\begin{equation}
D_j\frac{\partial^4w_j}{\partial x^4}+m_j\frac{\partial^2w_j}{\partial t^2}=q_j ,
\label{eq:plateeq}
\end{equation}
where $q_j$ is the net transverse load per unit area acting in $+y$. With $w_j\propto
e^{\ii(\omega t-k_xx)}$ the operator on the left reduces to multiplication by
\begin{equation}
G_j\equiv D_jk_x^4-m_j\omega^2 ,
\label{eq:Gdef}
\end{equation}
which vanishes at the in-vacuo flexural wavenumber of plate $j$. In non-dimensional variables
$G_1=m_1\omega^2S_1$ and $G_2=m_1\omega^2S_2$, which is the origin of the operators~\eqref{eq:Sdef}.

Two conditions are imposed at each wall.

\paragraph{(a) Kinematic condition.} The fluid must remain in contact with the plate, so the normal
fluid velocity equals the normal velocity of the plate surface following the mean flow,
\begin{equation}
v'=\frac{\mathrm{D}w_j}{\mathrm{D}t}=\partial_tw_j+U\partial_xw_j .
\label{eq:kinematic}
\end{equation}
For a compliant wall beneath a moving fluid the second term is the convective (Ingard--Myers)
contribution~\cite{ingard,myers}, and it is the source of the well-posedness difficulties recalled in
Section~\ref{sec:intro}. Here it vanishes identically: the condition is applied at $y=0$ and $y=a$,
where the Poiseuille profile satisfies no-slip, $U=0$. Hence
\begin{equation}
v'(0)=\ii\omega w_1,\qquad v'(a)=\ii\omega w_2
\label{eq:kin2}
\end{equation}
at every Mach number, exactly as in the quiescent problem.

\begin{remark}[Why the Ingard--Myers correction is absent]\label{rem:IM}
The Ingard--Myers condition is obtained by collapsing a thin mean boundary layer onto the wall, which
introduces a slip velocity at the wall and hence a convective term proportional to $U(\text{wall})$.
A fully developed Poiseuille profile has no such thin layer: the shear occupies the whole channel and
the mean velocity is exactly zero at the wall. The convective term is therefore not small, it is
identically zero, and no modified or regularised wall condition is required. This is a structural
feature of the configuration, and it is what allows the boundary conditions to be $M$-independent and
the $O(M)$ analysis to be carried out in closed form.
\end{remark}

\paragraph{(b) Dynamic condition.} The fluid loads each plate with its pressure, with a sign set by
which side the fluid lies on. The fluid is above plate~1 and below plate~2, so
\begin{equation}
q_1=-p'(x,0),\qquad q_2=+p'(x,a)
\quad\Longrightarrow\quad
G_1w_1=-p'(x,0),\qquad G_2w_2=+p'(x,a).
\label{eq:dynamic}
\end{equation}
Eliminating $w_j$ between~\eqref{eq:kin2} and~\eqref{eq:dynamic} gives the plate-admittance relations
\begin{equation}
\boxed{\;v'(0)=-\frac{\ii\omega}{G_1}\,p'(0),\qquad
v'(a)=+\frac{\ii\omega}{G_2}\,p'(a).\;}
\label{eq:adm}
\end{equation}
These are the only place the structure enters. A plate is, from the point of view of the fluid, a wall
of frequency- and wavenumber-dependent admittance $\ii\omega/G_j$; the admittance is infinite at the
in-vacuo flexural resonance ($G_j=0$, i.e.\ $S_j=0$, i.e.\ coincidence) and tends to zero as the plate
becomes rigid.

\subsection{Reduction map}
It is useful to record the limits in which the model collapses onto known problems:
\begin{center}
\begin{tabular}{@{}ll@{}}
\toprule
Limit & Reduces to\\
\midrule
$D_2\to\infty$ (plate 2 rigid) & single fluid-loaded plate over a rigid wall~\cite{sarkarsonti}\\
$D_2\to\infty$ and $D_1\to\infty$ & rigid-walled duct; cut-ons $\tau=n\pi$\\
$D_j\to0$ & pressure-release walls\\
$r_D=r_m=1$ & identical plates; symmetric/antisymmetric decoupling\\
$M\to0$ & quiescent problem, Section~\ref{sec:quiescent}\\
$\eps\to0$ & uncoupled structural and acoustic waves\\
\bottomrule
\end{tabular}
\end{center}

\section{The quiescent coupled dispersion relation}\label{sec:quiescent}

\subsection{Derivation}
With $M=0$ the mean flow is absent, $\sigma\equiv\omega$, and~\eqref{eq:leex}--\eqref{eq:leec} collapse
to the two-dimensional Helmholtz equation $p'_{xx}+p'_{yy}+k^2p'=0$ with $k=\omega/c$. Writing the
general $x$-travelling solution as a superposition of the two transverse-travelling components,
\begin{equation}
p'(x,y)=\bigl(Ae^{-\ii k_yy}+Be^{+\ii k_yy}\bigr)e^{-\ii k_xx},\qquad k_x^2+k_y^2=k^2,
\label{eq:pAB}
\end{equation}
the linearised $y$-momentum equation $\ii\omega\rho v'=-\partial_yp'$ gives the transverse velocity
\begin{equation}
v'(x,y)=\frac{k_y}{\omega\rho}\bigl(Ae^{-\ii k_yy}-Be^{+\ii k_yy}\bigr)e^{-\ii k_xx}.
\label{eq:vAB}
\end{equation}
Substituting~\eqref{eq:pAB} and~\eqref{eq:vAB} into the two admittance
relations~\eqref{eq:adm} yields a homogeneous $2\times2$ system for $(A,B)$. Introducing the
dimensionless fluid-loading admittance of each plate,
\begin{equation}
\alpha_j\equiv\frac{k_yG_j}{\omega^2\rho}=\frac{k_y\bigl(D_jk_x^4-m_j\omega^2\bigr)}{\omega^2\rho},
\label{eq:alphadef}
\end{equation}
the system reads
\begin{equation}
\begin{pmatrix}
1+\ii\alpha_1 & 1-\ii\alpha_1\\[2pt]
(1-\ii\alpha_2)e^{-\ii k_ya} & (1+\ii\alpha_2)e^{+\ii k_ya}
\end{pmatrix}
\begin{pmatrix}A\\B\end{pmatrix}=\begin{pmatrix}0\\0\end{pmatrix},
\end{equation}
and a non-trivial acoustic field exists only if the determinant vanishes. This gives the coupled
dispersion relation
\begin{equation}
\boxed{\;(1+\ii\alpha_1)(1+\ii\alpha_2)=e^{2\ii k_ya}\,(1-\ii\alpha_1)(1-\ii\alpha_2).\;}
\label{eq:disp}
\end{equation}
Because $(1+\ii\alpha)/(1-\ii\alpha)=e^{2\ii\arctan\alpha}$, equation~\eqref{eq:disp} is equivalent to
the phase-matching statement that a transverse round trip must return the wave to itself,
\begin{equation}
k_ya=\arctan\alpha_1+\arctan\alpha_2+n\pi
\qquad\Longleftrightarrow\qquad
\tan(k_ya)=\frac{\alpha_1+\alpha_2}{1-\alpha_1\alpha_2},
\label{eq:disptan}
\end{equation}
which is the most convenient working form: for real $k_y$ it is a single real transcendental equation
in $k_x$ valid for any combination of plate properties. The equivalence
of~\eqref{eq:disp} and the determinant was verified symbolically; the determinant equals the
left-minus-right side of~\eqref{eq:disp} multiplied by the non-vanishing factor $e^{-\ii k_ya}$, so no
roots are created or destroyed.

\subsection{Two checks}
\paragraph{Rigid-wall limit.} Letting plate~2 become rigid, $D_2\to\infty$ and hence
$\alpha_2\to\infty$, equation~\eqref{eq:disptan} gives $\tan(k_ya)=-1/\alpha_1$, i.e.
$\alpha_1=-\cot(k_ya)$, that is
\begin{equation}
D_1k_x^4-m_1\omega^2=-\frac{\omega^2\rho}{k_y}\cot(k_ya),
\label{eq:ss9}
\end{equation}
which is precisely the classical dispersion relation for a single fluid-loaded plate backed by a rigid
wall~\cite{sarkarsonti}.

\paragraph{Identical plates: exact symmetry factorisation.} If the plates are identical,
$\alpha_1=\alpha_2=\alpha$, the left- and right-hand sides of~\eqref{eq:disp} are perfect squares and
the relation factorises exactly,
\begin{equation}
\Bigl[(1+\ii\alpha)-e^{\ii k_ya}(1-\ii\alpha)\Bigr]\cdot
\Bigl[(1+\ii\alpha)+e^{\ii k_ya}(1-\ii\alpha)\Bigr]=0,
\label{eq:factor}
\end{equation}
giving two decoupled scalar problems:
\begin{equation}
\underbrace{\alpha=-\cot\frac{k_ya}{2}}_{\text{symmetric: }v'=0\text{ at }y=a/2}
\qquad\text{and}\qquad
\underbrace{\alpha=+\tan\frac{k_ya}{2}}_{\text{antisymmetric: }p'=0\text{ at }y=a/2}.
\label{eq:symanti}
\end{equation}
The symmetric family has a velocity node at the mid-plane and therefore behaves as a single flexible
plate facing a rigid wall placed at $y=a/2$; the two plates move in antiphase, $w_2=-w_1$, so the
channel breathes. The antisymmetric family has a pressure node at the mid-plane, the plates move in
phase, $w_2=+w_1$, and the channel sloshes. The symmetric relation is the single-plate
relation~\eqref{eq:ss9} with $a$ replaced by $a/2$; the antisymmetric relation has no single-plate
counterpart, which is why it is invisible in the one-flexible-plate geometry.

\subsection{Non-dimensional form and the uncoupled skeleton}
Applying the scaling~\eqref{eq:nondim} to~\eqref{eq:alphadef} gives
\begin{equation}
\alpha_1=\frac{1}{\eps}S_1(\xi)\,\tau,\qquad
\alpha_2=\frac{1}{\eps}S_2(\xi)\,\tau,\qquad
\tau=l\sqrt{\Om^2-\xi^2},
\label{eq:alphand}
\end{equation}
so that the two identical-plate relations~\eqref{eq:symanti} take the compact form
\begin{equation}
\text{(SYM)}\quad S(\xi)\,\tau\tan\frac{\tau}{2}+\eps=0,
\qquad\qquad
\text{(ANTI)}\quad S(\xi)\,\tau\cot\frac{\tau}{2}-\eps=0,
\qquad S=\frac{\xi^4}{\Om^2}-1 .
\label{eq:nd}
\end{equation}
For comparison, the single-plate relation of~\cite{sarkarsonti} in the same variables is
$S(\xi)\tau\tan\tau+\eps=0$; the symmetric family is obtained from it by the half-gap substitution
$\tan\tau\to\tan(\tau/2)$, exactly as it must be.

Setting $\eps=0$ in~\eqref{eq:nd} switches off the coupling and exposes the uncoupled parents from
which every coupled branch grows:
\begin{itemize}[leftmargin=1.6em,itemsep=1pt]
\item \textbf{Structural wave}, from $S(\xi)=0$: $\;\xi=\Om^{1/2}$. This is the in-vacuo flexural wave
and it belongs to \emph{both} families, corresponding to in-phase and out-of-phase plate flexure.
\item \textbf{Plane wave}, from $\tau=0$ in the symmetric relation: $\;\xi=\Om$. It belongs to the
symmetric family only.
\item \textbf{Even cut-ons}, from $\tan(\tau/2)=0$: $\;\tau=2n\pi$, symmetric family.
\item \textbf{Odd cut-ons}, from $\cot(\tau/2)=0$: $\;\tau=(2n{+}1)\pi$, antisymmetric family. The
antisymmetric family has \emph{no} plane wave, since $\tau\cot(\tau/2)\to2\ne0$ as $\tau\to0$.
\end{itemize}
The single rigid-duct cut-on ladder $\tau=n\pi$ of the one-plate problem therefore splits into two
interleaved ladders, one per symmetry family, and the full two-plate spectrum is twice as dense as the
one-plate spectrum at the same gap.

\subsection{Quiescent mode shapes}
The quiescent pressure eigenfunction satisfies $p_0''+\tau_0^2p_0=0$ on $0\le\eta\le1$ with
$\tau_0=l\sqrt{\Om^2-\xi_0^2}$, whence
\begin{equation}
p_0(\eta)=A\cos(\tau_0\eta)+B\sin(\tau_0\eta),
\label{eq:p0}
\end{equation}
with $B/A$ fixed by either boundary condition and $\xi_0$ by the dispersion relation. For the branches
used below,
\begin{equation}
\text{structural: }\xi_0=\Om^{1/2},\ \tau_0=l\sqrt{\Om^2-\Om};\qquad
\text{plane wave: }\xi_0=\Om,\ \tau_0=0,\ p_0=\text{const};
\end{equation}
\begin{equation}
\text{$m$-th cut-on: }\tau_0=m\pi,\ \xi_0=\sqrt{\Om^2-(m\pi/l)^2},
\end{equation}
with $m$ even for the symmetric family and odd for the antisymmetric family. These eigenpairs are the
input to every mean-flow formula in Sections~\ref{sec:lowmach}--\ref{sec:cl}; the quadratures they
require are collected in \ref{app:integrals}.

\section{The quiescent asymptotic catalogue}\label{sec:catalogue}

This section assembles, for both symmetry families, the closed-form expansions of the coupled
wavenumbers in the fluid-loading parameter. As stated in Section~\ref{sec:contrib}, the results
recover for the two-plate geometry the asymptotic-tracking framework
of~\cite{sarkarsonti,kunte2011,prakash2016} and are included as the $M=0$ baseline; the derivations are
sketched here and given in full in \ref{app:catalogue}. Throughout,
$\tau_s\equiv l\sqrt{\Om^2-\Om}$ denotes the transverse wavenumber evaluated on the uncoupled
structural wave.

\subsection{Small fluid loading, $\eps\ll1$}
Each coupled wavenumber is written as its uncoupled parent plus a correction,
$\xi=\xi^{(0)}+a_1\eps+O(\eps^2)$, and the powers of $\eps$ in~\eqref{eq:nd} are balanced.

\paragraph{Structural wave, both families.} Perturbing about $\xi^{(0)}=\Om^{1/2}$, where
$S(\xi)\approx4\Om^{-1/2}(\xi-\Om^{1/2})$, gives
\begin{equation}
\xi_s^{\rm sym}(\eps)=\Om^{1/2}-\frac{\eps}{4l\sqrt{\Om-1}\,\tan(\tau_s/2)},
\qquad
\xi_s^{\rm anti}(\eps)=\Om^{1/2}+\frac{\eps\,\tan(\tau_s/2)}{4l\sqrt{\Om-1}} .
\label{eq:struct}
\end{equation}
Below coincidence ($\Om<1$) one writes $\sqrt{\Om-1}\to\ii\sqrt{1-\Om}$ and $\tan\to\ii\tanh$, and both
corrections become real and positive: the fluid acts as an added mass and pushes the coupled
wavenumber above its in-vacuo value. The symmetric expression coincides with the single-plate result
of~\cite{sarkarsonti} under the half-gap substitution $l\to l/2$, an analytic consistency check.

\paragraph{Plane wave, symmetric family only.}
\begin{equation}
\xi_a^{\rm sym}(\eps)=\Om+\frac{\eps}{\Om l^2(\Om^2-1)} .
\label{eq:planewave}
\end{equation}
Above coincidence the correction is positive, placing the coupled plane wave slightly above $\xi=\Om$,
that is in the transversely evanescent region.

\paragraph{Cut-on branches, both families.} With
$\xi_0=\sqrt{\Om^2-(m\pi/l)^2}$ and $S_0=\xi_0^4/\Om^2-1$,
\begin{equation}
\xi_a^{(m)}(\eps)=\xi_0+\frac{2\eps}{S_0l^2\xi_0},
\label{eq:cuton}
\end{equation}
with $m$ even for the symmetric family and odd for the antisymmetric family. The factor of two relative
to the single-plate expression is the half-angle of the two-plate problem.

\subsection{Local $\sqrt\eps$ expansions at the crossings}
The regular expansions above fail wherever two uncoupled parents approach one another, because the
denominators in~\eqref{eq:struct}--\eqref{eq:cuton} vanish there. In an $O(\eps)$-wide neighbourhood of
each such point the correct balance is $O(\sqrt\eps)$, obtained by rescaling
$\xi=\xi^{(0)}+\sqrt\eps\,X$ and retaining the leading two terms. Two cases arise.

\paragraph{Symmetric coincidence, $\Om=1$: structure meets plane wave.}
\begin{equation}
\xi(\eps)=\Om\pm\frac{\sqrt\eps}{2l},\qquad\text{gap width }=\frac{\sqrt\eps}{l}.
\label{eq:gapcoinc}
\end{equation}
The upper branch above $\Om=1$ continues the structural branch from below and vice versa; the two
branches exchange identity across the gap.

\paragraph{Cut-on crossings, both families.} Where the structural wave meets the $m$-th cut-on, at the
frequency $\Om_c$ satisfying $l^2(\Om_c^2-\Om_c)=(m\pi)^2$ and hence $\xi_c=\sqrt{\Om_c}$,
\begin{equation}
\xi(\eps)=\sqrt{\Om_c}\pm\frac{\sqrt{2\eps}}{2l},\qquad\text{gap width }=\frac{\sqrt{2\eps}}{l}.
\label{eq:gapcuton}
\end{equation}

\paragraph{Antisymmetric coincidence: no gap.} The antisymmetric structural correction
in~\eqref{eq:struct} remains finite as $\Om\to1$, with $a_1\to\tfrac18$, so
$\xi\approx\Om^{1/2}(1+\eps/8)$ and the antisymmetric structural wave passes through coincidence
smoothly. The reason is structural rather than accidental: a gap requires two parents to collide, and
the antisymmetric family contains no plane wave for the structural wave to collide with at $\Om=1$.

\subsection{Heavy fluid loading, $\eps\gg1$}
Setting $\eps'=1/\eps\ll1$ and repeating the procedure, the parents are no longer the rigid-duct
cut-ons but the \emph{pressure-release} cut-ons, and for both families
\begin{equation}
\xi'^{(m)}_a(\eps)=\sqrt{\Om^2-\Bigl(\frac{m\pi}{l}\Bigr)^2}-\frac{2(m\pi)^2S_0}{l^2\xi_0}\,\frac1\eps,
\label{eq:heavy}
\end{equation}
now with $m$ \emph{odd} for the symmetric family and \emph{even} for the antisymmetric family, the
latter including an $m=0$ pressure-release mode. The parities are therefore exchanged relative to the
light-loading case, and each branch migrates steadily from a rigid-walled to a pressure-release
wavenumber as the loading grows. Table~\ref{tab:migration} records the correspondence.

\begin{table}[htb]
\centering
\small
\begin{tabular}{@{}lll@{}}
\toprule
Family & Light loading, $\eps\ll1$ & Heavy loading, $\eps\gg1$\\
 & (rigid-walled parents) & (pressure-release parents)\\
\midrule
Symmetric & $\tau=2n\pi$ (even) $+$ plane wave $\xi=\Om$ & $\tau=(2m{+}1)\pi$ (odd)\\
Antisymmetric & $\tau=(2n{+}1)\pi$ (odd) & $\tau=2m\pi$ (even), including $m=0$\\
\bottomrule
\end{tabular}
\caption{Parity exchange of the two cut-on ladders between light and heavy fluid loading. Each family
migrates from the rigid-walled to the pressure-release ladder; because the two ladders interleave, the
full two-plate spectrum is twice as dense as the single-plate spectrum.}
\label{tab:migration}
\end{table}

\subsection{Accuracy of the catalogue}
Every expression above was checked against the exact numerical roots of~\eqref{eq:nd}. The maximum
relative errors within the stated ranges of validity are collected in Table~\ref{tab:cat}; the
composite comparison is shown in Fig.~\ref{fig:catalogue}. Away from the transition windows the
first-order expressions are accurate to better than half a percent at $\eps=0.05$, and the local
$\sqrt\eps$ forms reproduce the gap edges to better than $0.06\%$ at $\eps=0.03$.

\begin{table}[htb]
\centering
\resizebox{\textwidth}{!}{%
\begin{tabular}{@{}lllll@{}}
\toprule
Region & Family & Closed form & Order & Max.\ error\\
\midrule
Structural, away from crossings & sym & $\Om^{1/2}-\eps/[4l\sqrt{\Om-1}\tan(\tau_s/2)]$ & $\eps$ & $0.08$--$0.49\%$\\
Structural, away from crossings & anti & $\Om^{1/2}+\eps\tan(\tau_s/2)/[4l\sqrt{\Om-1}]$ & $\eps$ & $0.009$--$0.22\%$\\
Plane wave, away from $\Om=1$ & sym & $\Om+\eps/[\Om l^2(\Om^2-1)]$ & $\eps$ & $0.016$--$0.33\%$\\
Cut-on, away from crossings & both & $\xi_0+2\eps/(S_0l^2\xi_0)$, $m$ even/odd & $\eps$ & $0.08$--$1\%$\\
Coincidence & sym & $\Om\pm\sqrt\eps/(2l)$ & $\sqrt\eps$ & $<0.03\%$\\
Coincidence & anti & regular, $\Om^{1/2}(1+\eps/8)$ & $\eps$ & ---\\
Cut-on crossing & both & $\sqrt{\Om_c}\pm\sqrt{2\eps}/(2l)$ & $\sqrt\eps$ & $<0.06\%$\\
Heavy loading & both & $\xi_0-2(m\pi)^2S_0\eps^{-1}/(l^2\xi_0)$, $m$ odd/even & $\eps^{-1}$ & $0.1$--$3.9\%$\\
\bottomrule
\end{tabular}}
\caption{The quiescent asymptotic catalogue and its validated accuracy. Errors are maxima of the
relative deviation from the exact numerical root over the stated range, at $l=3$ with $\eps=0.05$
(regular expansions), $\eps=0.03$ (local $\sqrt\eps$ expansions) and $\eps=20$ (heavy loading).}
\label{tab:cat}
\end{table}

\begin{figure}[htb]
\centering
\includegraphics[width=0.96\textwidth]{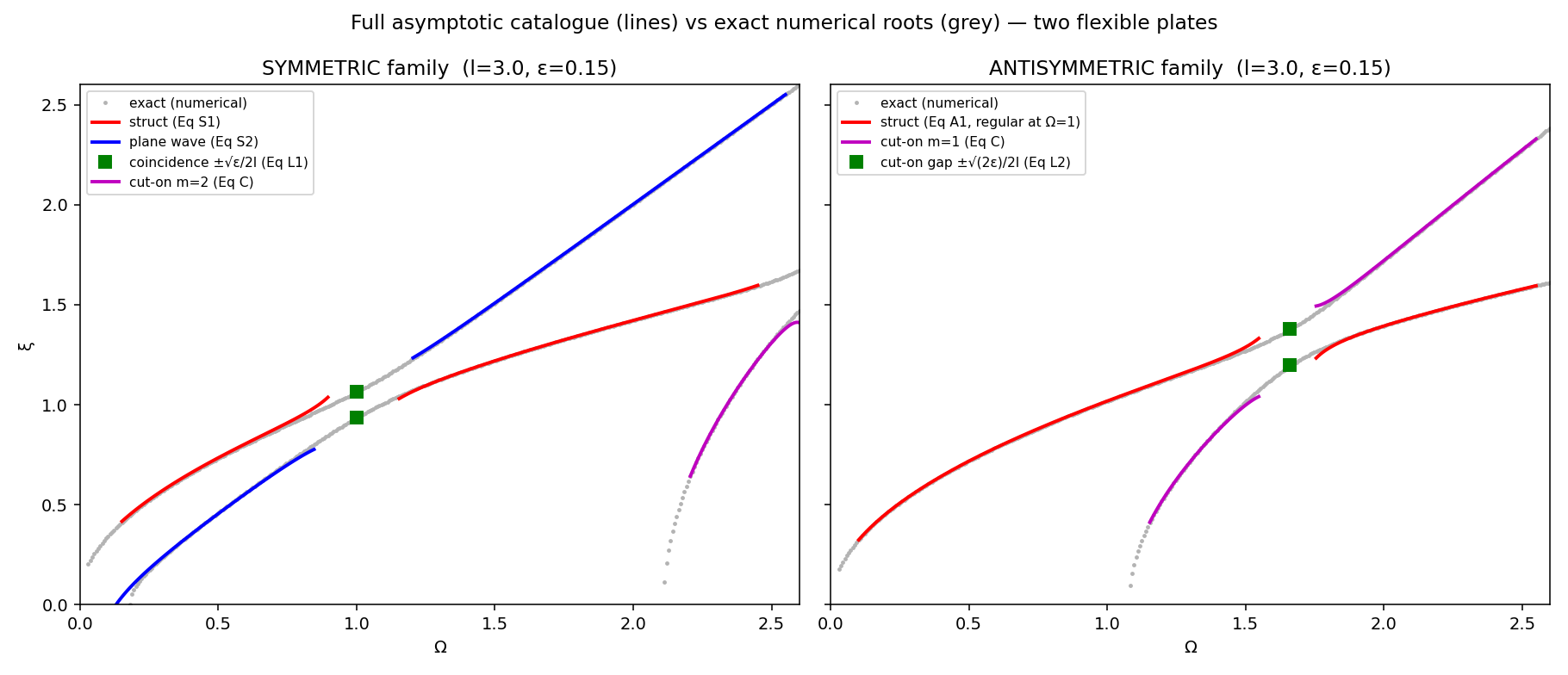}
\caption{The complete quiescent asymptotic catalogue (coloured lines and green gap markers) overlaid on
the exact numerical roots (grey), for $l=3$, $\eps=0.15$. Left: symmetric family, showing the
structural branch~\eqref{eq:struct}, the plane wave~\eqref{eq:planewave}, the $m=2$
cut-on~\eqref{eq:cuton} and the coincidence gap~\eqref{eq:gapcoinc}. Right: antisymmetric family,
showing the structural branch, the $m=1$ cut-on and the cut-on gap~\eqref{eq:gapcuton}; note the
absence of a plane wave and the smooth passage of the structural branch through $\Om=1$. A single
physical dispersion curve is a patchwork of these forms: on the symmetric side, plane wave
$\to$ (coincidence gap) $\to$ structural $\to$ (cut-on gap) $\to$ cut-on.}
\label{fig:catalogue}
\end{figure}

\begin{figure}[htb]
\centering
\includegraphics[width=0.9\textwidth]{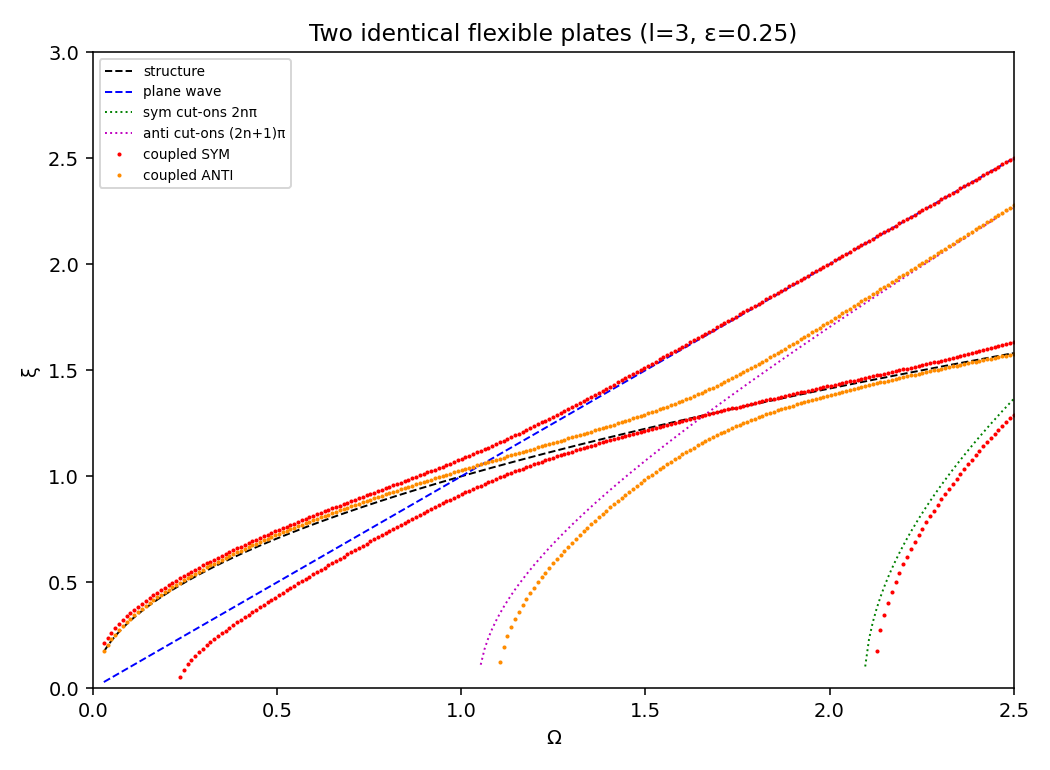}
\caption{Coupled symmetric and antisymmetric dispersion branches for two identical flexible plates
($l=3$, $\eps=0.25$), overlaid on the uncoupled parents of~\eqref{eq:nd}. The symmetric family tracks
the plane wave and the even cut-ons; the antisymmetric family tracks the odd cut-ons and has no
plane-wave member. Both families perturb the common structural wave $\xi=\Om^{1/2}$, and a gap opens at
every intersection.}
\label{fig:branches}
\end{figure}

\subsection{Non-identical plates: veering}
When $r_D\ne1$ or $r_m\ne1$ the factorisation~\eqref{eq:factor} fails and the full
relation~\eqref{eq:disptan} must be used. The two structural waves, in-phase and out-of-phase plate
flexure, then have different in-vacuo wavenumbers, and where they would otherwise cross they instead
veer, exchanging mode shape in the manner familiar from two coupled oscillators. Perturbing about the
identical-plate solution gives a minimum separation
\begin{equation}
\Delta\xi_{\rm veer}=\frac{\eps}{2l\sqrt{\Om-1}\,\sin\tau_s}
\qquad\text{at }r_D=r_m=1,
\label{eq:veer}
\end{equation}
which grows linearly with the fluid loading, as confirmed numerically. Physically, asymmetry in
stiffness or mass breaks the mirror symmetry of the channel, so symmetric and antisymmetric motions
cease to be eigenstates and exchange energy. Adding a structural loss factor,
$D_j\to D_j(1+\ii\eta_j)$, renders every wavenumber complex; the structural branches attenuate most and
the plane wave least, since the former carry most of their energy in the plates.

\begin{figure}[htb]
\centering
\includegraphics[width=0.98\textwidth]{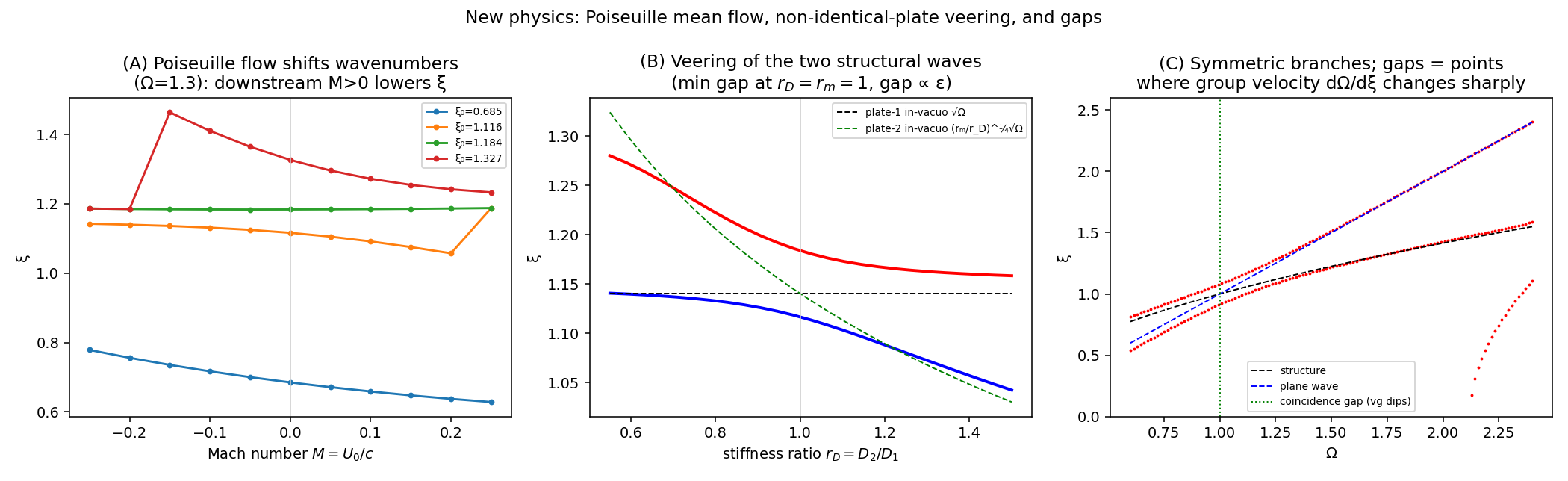}
\caption{(a) Poiseuille mean flow shifts every coupled wavenumber, and the shift is odd in $M$
(anticipating Section~\ref{sec:lowmach}); $\Om=1.3$. (b) Veering of the two structural waves as the
stiffness ratio $r_D=D_2/D_1$ is varied through unity, with the in-vacuo parents shown dashed; the
minimum gap is~\eqref{eq:veer} and is proportional to $\eps$. (c) Symmetric coupled branches, with the
coincidence gap marked; the group velocity $\dd\Om/\dd\xi$ changes sharply across each $\sqrt\eps$ gap,
which is the signature of energy switching carrier between the structure and the fluid.}
\label{fig:newphysics}
\end{figure}

\section{The Pridmore--Brown equation with plate-admittance conditions}\label{sec:pb}

\subsection{Why a new wave equation is needed}
With a uniform mean flow $U=\text{const}$ the wave operator transforms under a Galilean shift
$\omega\to\omega-k_xU$ and nothing structural changes. A sheared flow $U=U(y)$ destroys this. The term
$\rho U'v'$ in~\eqref{eq:leex} couples the transverse velocity into the axial momentum balance, so the
three perturbation fields no longer separate, and eliminating $u'$ and $v'$ in favour of $p'$ produces
an ordinary differential equation with variable coefficients. That equation is the Pridmore--Brown
equation~\cite{pridmorebrown}.

\subsection{Derivation}
Under the wave ansatz, \eqref{eq:leex}--\eqref{eq:leec} become, with $\sigma$ as
in~\eqref{eq:sigmadef},
\begin{align}
\ii\sigma\rho u'&=\ii k_xp'-\rho U'v',
\label{eq:pbu}\\
\ii\sigma\rho v'&=-\partial_yp',
\label{eq:pbv}\\
\ii\sigma p'+\rho c^2\bigl(-\ii k_xu'+\partial_yv'\bigr)&=0 .
\label{eq:pbc}
\end{align}

\emph{Step 1.} Solve~\eqref{eq:pbv} for the transverse velocity:
\begin{equation}
v'=\frac{\ii\,\partial_yp'}{\sigma\rho}.
\label{eq:vfromp}
\end{equation}

\emph{Step 2.} Substitute into~\eqref{eq:pbu} and solve for the axial velocity:
\begin{equation}
\ii\sigma\rho u'=\ii k_xp'-\rho U'\frac{\ii\,\partial_yp'}{\sigma\rho}
\quad\Longrightarrow\quad
u'=\frac{k_xp'}{\sigma\rho}+\frac{U'\,\partial_yp'}{\ii\sigma^2\rho}.
\label{eq:ufromp}
\end{equation}
The second term is the shear contribution; it is absent for uniform flow.

\emph{Step 3.} Substitute~\eqref{eq:vfromp} and~\eqref{eq:ufromp} into the continuity
relation~\eqref{eq:pbc}:
\begin{equation}
-\ii\sigma p'+\rho c^2\left[-\ii k_x\left(\frac{k_xp'}{\sigma\rho}
+\frac{U'\partial_yp'}{\ii\sigma^2\rho}\right)
+\partial_y\!\left(\frac{\ii\,\partial_yp'}{\sigma\rho}\right)\right]=0 .
\end{equation}
Expanding the last term with $\partial_y\sigma=-k_xU'$ gives
$\partial_y(\ii\partial_yp'/\sigma\rho)=\ii\partial_{yy}p'/(\sigma\rho)+\ii k_xU'\partial_yp'/(\sigma^2\rho)$,
so that
\begin{equation}
-\ii\sigma p'-\frac{\ii k_x^2c^2p'}{\sigma}
+\frac{k_xc^2U'\partial_yp'}{\sigma^2}
+\frac{\ii c^2\partial_{yy}p'}{\sigma}
+\frac{\ii k_xc^2U'\partial_yp'}{\sigma^2}=0 .
\end{equation}
Multiplying by $\sigma/(\ii c^2)$ and combining the two identical shear contributions (one from the
axial momentum equation and one from the continuity cross-term, which together supply the factor of two)
yields the dimensional Pridmore--Brown equation
\begin{equation}
\partial_{yy}p'+\frac{2k_xU'}{\sigma}\,\partial_yp'
+\left(\frac{\sigma^2}{c^2}-k_x^2\right)p'=0 .
\label{eq:pbdim}
\end{equation}

\emph{Step 4.} Non-dimensionalise with $\eta=y/a$, $\partial_y=a^{-1}\partial_\eta$,
$\sigma=\omega_c\sig$, $U'=U_0f'/a$, $M=U_0/c$, $k_x=\xi\omega_c/c$ and $l=\omega_ca/c$:
\begin{equation}
\boxed{\;
\hat p_{\eta\eta}+\frac{2M\xi f'(\eta)}{\sig(\eta)}\,\hat p_\eta
+l^2\bigl(\sig^2-\xi^2\bigr)\hat p=0,
\qquad \sig(\eta)=\Om-M\xi f(\eta),\quad f'=4(1-2\eta).\;}
\label{eq:pb}
\end{equation}

\begin{remark}[Sign of the shear term]
The sign of the middle coefficient is the one point in the derivation that is easy to get wrong. It was
verified independently by a primitive-variable formulation (the three-field system~\eqref{eq:pbu}--\eqref{eq:pbc} solved without
eliminating $u'$ and $v'$ a priori, then reduced symbolically to a single equation in $\hat p$),
which reproduces~\eqref{eq:pb} exactly, including the sign, in the
$e^{\ii(\omega t-k_xx)}$, $\sig=\Om-M\xi f$ convention used here. See Section~\ref{sec:val}.
\end{remark}

\begin{remark}[Reduction at $M=0$]
At $M=0$ one has $\sig=\Om$, the shear term vanishes identically, and~\eqref{eq:pb} collapses to
$\hat p_{\eta\eta}+l^2(\Om^2-\xi^2)\hat p=0$, the quiescent Helmholtz equation whose solution
is~\eqref{eq:p0} and whose eigenvalue condition is the coupled dispersion relation~\eqref{eq:disp}.
\end{remark}

\subsection{Boundary conditions}
The admittance relations~\eqref{eq:adm} are converted into conditions on $\hat p$ using the
inviscid $y$-momentum equation~\eqref{eq:vfromp}. At $\eta=0$, where no-slip gives $\sigma=\omega$,
\begin{equation}
v'(0)=\frac{\ii\,\partial_yp'(0)}{\omega\rho}
\;\overset{\eqref{eq:adm}}{=}\;-\frac{\ii\omega}{G_1}p'(0)
\quad\Longrightarrow\quad
\partial_yp'(0)=-\frac{\omega^2\rho}{G_1}p'(0),
\end{equation}
and non-dimensionalising with $\hat p=p'/\rho c^2$, $\eta=y/a$, $G_1=m_1\omega^2S_1$ and
$\eps=\rho a/m_1$ gives, together with its counterpart at $\eta=1$,
\begin{equation}
\boxed{\;\hat p_\eta(0)+\frac{\eps}{S_1}\hat p(0)=0,
\qquad
\hat p_\eta(1)-\frac{\eps}{S_2}\hat p(1)=0.\;}
\label{eq:robinp}
\end{equation}
These are Robin (third-kind) conditions, and the ratio $\eps/S_j$ is the plate fluid-loading admittance,
i.e.\ how much the plate deflects per unit applied pressure. Two features deserve emphasis. First,
by Remark~\ref{rem:IM} the conditions~\eqref{eq:robinp} are \emph{independent of $M$}: the mean flow
enters only through the bulk operator in~\eqref{eq:pb}. Second, the eigenvalue $\xi$ appears
nonlinearly in the boundary conditions through $S_j(\xi)$ as well as in the operator, so
\eqref{eq:pb}--\eqref{eq:robinp} constitute a nonlinear eigenvalue problem in $\xi$. This nonlinearity
is what produces the boundary contributions to the solvability condition in
Section~\ref{sec:solvability}, and it is what distinguishes a plate-bounded waveguide from a
constant-impedance lined duct.

\section{Closed-form first-order Mach correction}\label{sec:lowmach}

\subsection{The expansion and its domain of validity}
For $M$ below the critical onset value $M_{\rm crit}=\Om/\xi$ identified in Section~\ref{sec:cl}, the
coefficient $2M\xi f'/\sig$ in~\eqref{eq:pb} is smooth on $[0,1]$ because $\sig$ does not vanish there.
The eigenvalue and eigenfunction are then analytic in $M$ and may be expanded:
\begin{equation}
\xi=\xi_0+M\xi_1+M^2\xi_2+O(M^3),
\qquad
\hat p(\eta)=p_0(\eta)+Mp_1(\eta)+M^2p_2(\eta)+O(M^3),
\label{eq:Mexp}
\end{equation}
where every coefficient depends on $\Om$, $\eps$, $l$, $r_D$ and $r_m$ but not on $M$. The
leading pair $(\xi_0,p_0)$ is a quiescent eigenpair from Section~\ref{sec:quiescent}.

\subsection{Expansion of the operator}
Substituting~\eqref{eq:Mexp} into $\sig=\Om-M\xi f$ gives
\begin{equation}
\sig=\underbrace{\Om}_{\sig^{(0)}}
\underbrace{-M\xi_0f}_{\sig^{(1)}}
\underbrace{-M^2\xi_1f}_{\sig^{(2)}}+O(M^3),
\label{eq:sigexp}
\end{equation}
and the key simplification at first order is that $\sig^{(0)}=\Om$ is a \emph{constant}: the
leading-order Doppler shift is spatially uniform, so the leading operator retains constant
coefficients. Squaring,
\begin{align}
\sig^2&=\Om^2-2M\Om\xi_0f-M^2\bigl(2\Om\xi_1f-\xi_0^2f^2\bigr)+O(M^3),\\
\xi^2&=\xi_0^2+2M\xi_0\xi_1+M^2\bigl(\xi_1^2+2\xi_0\xi_2\bigr)+O(M^3),
\end{align}
so that
\begin{equation}
\sig^2-\xi^2=\underbrace{(\Om^2-\xi_0^2)}_{A_0}
+M\underbrace{\bigl(-2\Om\xi_0f-2\xi_0\xi_1\bigr)}_{A_1}
+M^2\underbrace{\bigl(\xi_0^2f^2-2\Om\xi_1f-\xi_1^2-2\xi_0\xi_2\bigr)}_{A_2}+O(M^3).
\label{eq:Aexp}
\end{equation}
The shear coefficient expands as
\begin{equation}
\frac{2M\xi f'}{\sig}
=\frac{2M\xi_0f'}{\Om}\cdot\frac{1}{1-M\xi_0f/\Om+\cdots}+O(M^3)
=\frac{2M\xi_0f'}{\Om}+\frac{2M^2\xi_0^2ff'}{\Om^2}+O(M^3).
\label{eq:shearexp}
\end{equation}

\subsection{Order $M^0$: the quiescent problem}
Collecting terms of order $M^0$ in~\eqref{eq:pb} with~\eqref{eq:Aexp}--\eqref{eq:shearexp} recovers
\begin{equation}
L_0[p_0]\equiv p_0''+\tau_0^2p_0=0,\qquad \tau_0^2=l^2(\Om^2-\xi_0^2),
\label{eq:L0}
\end{equation}
subject to the boundary conditions~\eqref{eq:robinp} evaluated at $\xi_0$,
\begin{equation}
p_0'(0)+\frac{\eps}{S_1^{(0)}}p_0(0)=0,
\qquad
p_0'(1)-\frac{\eps}{S_2^{(0)}}p_0(1)=0,
\qquad S_j^{(0)}\equiv S_j(\xi_0).
\label{eq:bc0}
\end{equation}

\subsection{Order $M^1$: the forced problem}
At order $M^1$ the second derivative contributes $p_1''$; the shear term
contributes $2\xi_0f'p_0'/\Om$; and $l^2(\sig^2-\xi^2)\hat p$ contributes
$l^2A_0p_1+l^2A_1p_0$. Adding these,
\begin{equation}
L_0[p_1]=F(\eta;\xi_1),
\qquad
F\equiv-\frac{2\xi_0f'}{\Om}p_0'+2l^2\xi_0\bigl(\Om f+\xi_1\bigr)p_0 .
\label{eq:M1}
\end{equation}
The forcing has two physically distinct parts. The term $-2\xi_0f'p_0'/\Om$ arises from the spatial
variation of the mean flow and would be absent for a plug profile. The term
$2l^2\xi_0(\Om f+\xi_1)p_0$ contains the Doppler convection ($\Om f$) and the shift of the eigenvalue
itself ($\xi_1$). Crucially, the unknown scalar $\xi_1$ enters $F$ linearly, so the solvability
condition below is a linear equation for $\xi_1$.

The boundary conditions must also be expanded, because $S_j$ depends on $\xi$. Since
$\partial S_1/\partial\xi=4\xi^3/\Om^2$,
\begin{equation}
S_1(\xi)=S_1^{(0)}+M\xi_1\frac{4\xi_0^3}{\Om^2}+O(M^2)
\quad\Longrightarrow\quad
\frac{\eps}{S_1(\xi)}=\frac{\eps}{S_1^{(0)}}
-M\frac{4\eps\xi_0^3\xi_1}{\Om^2\bigl(S_1^{(0)}\bigr)^2}+O(M^2),
\end{equation}
so that collecting $O(M)$ terms in $\hat p_\eta(0)+(\eps/S_1)\hat p(0)=0$ and its counterpart gives the
\emph{inhomogeneous} first-order boundary conditions
\begin{equation}
p_1'(0)+\frac{\eps}{S_1^{(0)}}p_1(0)=\frac{4\eps\xi_0^3\xi_1}{\Om^2\bigl(S_1^{(0)}\bigr)^2}p_0(0),
\qquad
p_1'(1)-\frac{\eps}{S_2^{(0)}}p_1(1)=-\frac{4\eps r_D\xi_0^3\xi_1}{\Om^2\bigl(S_2^{(0)}\bigr)^2}p_0(1).
\label{eq:bc1}
\end{equation}
It is this inhomogeneity, a direct consequence of the plate admittance depending on the eigenvalue,
that produces the second term in the denominator of the final result.

\subsection{Self-adjointness and the Fredholm alternative}\label{sec:solvability}

\begin{theorem}[Self-adjointness]\label{thm:sa}
The operator $L_0=\dd^2/\dd\eta^2+\tau_0^2$, equipped with the homogeneous Robin
conditions~\eqref{eq:bc0}, is self-adjoint on $L^2[0,1]$.
\end{theorem}

\begin{proof}
For any $u,v\in C^2[0,1]$, integration by parts twice gives the Lagrange identity
\begin{equation}
\int_0^1v\,L_0[u]\,\dd\eta=\bigl[vu'-uv'\bigr]_0^1+\int_0^1u\,L_0[v]\,\dd\eta .
\label{eq:lagrange}
\end{equation}
If both $u$ and $v$ satisfy~\eqref{eq:bc0}, then at $\eta=1$ one has $u'(1)=(\eps/S_2^{(0)})u(1)$ and
$v'(1)=(\eps/S_2^{(0)})v(1)$, so $v(1)u'(1)-u(1)v'(1)=0$; at $\eta=0$ one has
$u'(0)=-(\eps/S_1^{(0)})u(0)$ and likewise for $v$, so $v(0)u'(0)-u(0)v'(0)=0$. The boundary term
in~\eqref{eq:lagrange} vanishes and $\langle v,L_0u\rangle=\langle L_0v,u\rangle$.
\end{proof}

Because $L_0$ is self-adjoint, its adjoint null function is $p_0$ itself, and the Fredholm alternative
states that~\eqref{eq:M1} possesses a solution if and only if the forcing, \emph{including the
inhomogeneous boundary data}, is orthogonal to $p_0$. Applying~\eqref{eq:lagrange} with $u=p_1$,
$v=p_0$ and using $L_0[p_0]=0$,
\begin{equation}
\int_0^1p_0L_0[p_1]\,\dd\eta=\bigl[p_0p_1'-p_0'p_1\bigr]_0^1
\qquad\Longrightarrow\qquad
\bigl[p_0p_1'-p_0'p_1\bigr]_0^1=\langle p_0,F\rangle .
\label{eq:solv}
\end{equation}

\subsection{Evaluating the two sides}

\paragraph{Interior integral.} Split $F$ as in~\eqref{eq:M1}. For the shear part, use
$p_0f'p_0'=f'\,\dd(p_0^2/2)/\dd\eta$ and integrate by parts:
\begin{equation}
\int_0^1f'\frac{\dd(p_0^2/2)}{\dd\eta}\,\dd\eta
=\left[\frac{f'p_0^2}{2}\right]_0^1-\frac12\int_0^1f''p_0^2\,\dd\eta .
\end{equation}
With $f'(0)=4$, $f'(1)=-4$ and $f''=-8$ from~\eqref{eq:fprops},
\begin{equation}
\left[\frac{f'p_0^2}{2}\right]_0^1=-2p_0(1)^2-2p_0(0)^2,
\qquad
-\frac12\int_0^1f''p_0^2\,\dd\eta=4I_0,
\qquad I_0\equiv\int_0^1p_0^2\,\dd\eta,
\end{equation}
so that
\begin{equation}
I_A\equiv-\frac{2\xi_0}{\Om}\int_0^1p_0f'p_0'\,\dd\eta
=\frac{4\xi_0}{\Om}\bigl(p_0(1)^2+p_0(0)^2\bigr)-\frac{8\xi_0}{\Om}I_0 .
\label{eq:IA}
\end{equation}

\begin{remark}[The role of the parabolic profile]\label{rem:fpp}
The integration by parts leaves the boundary term $[f'p_0^2/2]_0^1$ and the interior integral
$-\tfrac12\int_0^1f''p_0^2\,\dd\eta$. For a general smooth no-slip profile the latter is a third
elementary quadrature of the \emph{known} quiescent mode shape, so the closed form holds for any such
profile, not only the parabolic one. Poiseuille flow is distinguished merely in that $f''$ is constant:
the third quadrature then collapses onto the norm, $\int_0^1f''p_0^2\,\dd\eta=-8I_0$, and no integral
beyond $I_0$ and $I_f$ survives.
\end{remark}

The convection and eigenvalue part gives, with $I_f\equiv\int_0^1f\,p_0^2\,\dd\eta$,
\begin{equation}
I_B\equiv2l^2\xi_0\int_0^1p_0(\Om f+\xi_1)p_0\,\dd\eta
=2l^2\xi_0\bigl(\Om I_f+\xi_1I_0\bigr),
\label{eq:IB}
\end{equation}
so that
\begin{equation}
\langle p_0,F\rangle=\frac{4\xi_0}{\Om}\bigl(p_0(1)^2+p_0(0)^2\bigr)-\frac{8\xi_0}{\Om}I_0
+2l^2\xi_0\Om I_f+2l^2\xi_0\xi_1I_0 .
\label{eq:rhs}
\end{equation}

\paragraph{Boundary term.} Using the zeroth-order conditions~\eqref{eq:bc0} and the first-order
conditions~\eqref{eq:bc1}, at $\eta=1$
\begin{equation}
\begin{aligned}
p_0(1)p_1'(1)-p_0'(1)p_1(1)
&=p_0(1)\!\left[\frac{\eps}{S_2^{(0)}}p_1(1)
-\frac{4\eps r_D\xi_0^3\xi_1}{\Om^2\bigl(S_2^{(0)}\bigr)^2}p_0(1)\right]
-\frac{\eps}{S_2^{(0)}}\,p_0(1)p_1(1)\\
&=-\frac{4\eps r_D\xi_0^3\xi_1}{\Om^2\bigl(S_2^{(0)}\bigr)^2}\,p_0(1)^2,
\end{aligned}
\end{equation}
and, by the same steps at $\eta=0$,
\begin{equation}
p_0(0)p_1'(0)-p_0'(0)p_1(0)
=+\frac{4\eps\xi_0^3\xi_1}{\Om^2\bigl(S_1^{(0)}\bigr)^2}\,p_0(0)^2 .
\end{equation}
Hence
\begin{equation}
\bigl[p_0p_1'-p_0'p_1\bigr]_0^1
=-\frac{4\eps\xi_0^3\xi_1}{\Om^2}
\left[\frac{r_Dp_0(1)^2}{(S_2^{(0)})^2}+\frac{p_0(0)^2}{(S_1^{(0)})^2}\right].
\label{eq:lhs}
\end{equation}
Note that the unknown $p_1$ has cancelled identically from both endpoints; this is the practical
content of the solvability condition.

\subsection{The closed-form result}
Equating~\eqref{eq:lhs} and~\eqref{eq:rhs}, collecting the $\xi_1$ terms on the left and dividing
through by $-2\xi_0$, gives the first-order correction for an \emph{arbitrary} smooth no-slip shear
profile:
\begin{equation}
\xi_1=\dfrac{\;\dfrac{1}{2\Om}\bigl[f'(1)p_0(1)^2-f'(0)p_0(0)^2\bigr]-\dfrac{I_{f''}}{2\Om}-l^2\Om\,I_f\;}
{\;l^2I_0+\dfrac{2\eps\xi_0^2}{\Om^2}
\left(\dfrac{r_Dp_0(1)^2}{S_2^2}+\dfrac{p_0(0)^2}{S_1^2}\right)\;},
\qquad I_{f''}=\int_0^1 f''p_0^2\,\dd\eta .
\label{eq:xi1gen}
\end{equation}
It needs only the endpoint pressures $p_0(0),p_0(1)$, the endpoint profile slopes $f'(0),f'(1)$, and the
three elementary quadratures $\int p_0^2$, $\int f''p_0^2$ and $\int fp_0^2$ of the known quiescent mode
shape; no differential equation is solved. For the parabolic Poiseuille profile $f'(0)=4$, $f'(1)=-4$,
$f''=-8$, so the second quadrature collapses onto the norm and~\eqref{eq:xi1gen} reduces to the fully
elementary central result
\begin{equation}
\boxed{\;
\xi_1=\dfrac{\;-\dfrac{2}{\Om}\bigl(p_0(1)^2+p_0(0)^2\bigr)
+\dfrac{4}{\Om}I_0-l^2\Om\,I_f\;}
{\;l^2I_0+\dfrac{2\eps\xi_0^2}{\Om^2}
\left(\dfrac{r_Dp_0(1)^2}{S_2^2}+\dfrac{p_0(0)^2}{S_1^2}\right)\;},
\qquad
I_0=\int_0^1p_0^2\,\dd\eta,\quad I_f=\int_0^1f\,p_0^2\,\dd\eta .\;}
\label{eq:xi1}
\end{equation}
Everything on the right-hand side is a quiescent quantity: given any coupled quiescent eigenpair
$(\xi_0,p_0)$ from Section~\ref{sec:quiescent}, the endpoint values and quadratures (elementary for
the cosine-plus-sine mode shapes~\eqref{eq:p0}, see \ref{app:integrals}) give $\xi_1$
immediately. The result holds for arbitrary $\eps$, for non-identical plates, and on every branch.

\subsection{Sign of the correction}

\begin{proposition}\label{prop:sign}
$\xi_1<0$ on every branch.
\end{proposition}

\begin{proof}[Proof (plane-wave branch, and argument for the remainder)]
The denominator of~\eqref{eq:xi1} is a sum of squares with positive coefficients and is manifestly
positive. For the plane wave, $p_0\equiv1$, so $I_0=1$, $p_0(0)=p_0(1)=1$ and
$I_f=\int_0^14\eta(1-\eta)\dd\eta=4(\tfrac12-\tfrac13)=\tfrac23$, whence the numerator is
\begin{equation}
-\frac{4}{\Om}+\frac{4}{\Om}-l^2\Om\cdot\frac23=-\frac{2l^2\Om}{3}<0 .
\end{equation}
On the structural and cut-on branches the boundary-shear and bulk-shear contributions partially cancel
as they do here, while $I_f>0$ because $f\ge0$ on $[0,1]$; the Doppler term $-l^2\Om I_f$ therefore
dominates and the numerator remains negative. This is confirmed for every branch computed in
Table~\ref{tab:xi1} and Fig.~\ref{fig:newphysics}(a).
\end{proof}

Physically, a downstream mean flow increases the effective propagation speed of every mode (the
waves ride the current), and a higher phase speed at fixed frequency means a lower axial wavenumber
$\xi=k_x/k_c$. Correspondingly, an upstream flow ($M\to-M$) raises $\xi$ by the same amount at this
order. The dispersion relation is therefore no longer symmetric under $k_x\to-k_x$: the mean flow
breaks the left--right symmetry of the waveguide, and upstream- and downstream-propagating waves of the
same frequency have different wavenumbers.

\subsection{Term-by-term interpretation}
Each term in~\eqref{eq:xi1} has an identifiable physical origin.

\begin{itemize}[leftmargin=1.6em,itemsep=2pt]
\item $-\dfrac{2}{\Om}\bigl(p_0(1)^2+p_0(0)^2\bigr)$ , \emph{boundary shear}. The mean velocity
gradient at the walls, $f'(0)=4$ and $f'(1)=-4$, modifies the effective plate admittance. It is
largest for modes with high wall pressure, that is the structural and cut-on branches.

\item $+\dfrac{4}{\Om}I_0$ , \emph{bulk shear}. The uniform curvature $f''=-8$ of the profile acts
like a spatially distributed added mass throughout the column.

\item $-l^2\Om I_f$ , \emph{Doppler convection}. The mean flow carries the wave downstream. The
weighting $I_f=\int f p_0^2$ is largest for modes that concentrate their energy in the fast core, which
is why the plane wave ($p_0\equiv1$) is by a wide margin the most flow-sensitive branch, and why the
structural branches, whose speed is set mainly by plate stiffness, are the least sensitive.

\item $l^2I_0$ in the denominator, \emph{acoustic inertia} of the fluid column.

\item $\dfrac{2\eps\xi_0^2}{\Om^2}\bigl(\cdots\bigr)$ in the denominator, \emph{plate inertia},
entering through the eigenvalue-dependence of the plate admittance. It vanishes as $\eps\to0$, leaving
only the acoustic inertia, and it is singular at coincidence, where $S_j\to0$ and the plate becomes
infinitely compliant.
\end{itemize}

\subsection{Validation of $\xi_1$}
Table~\ref{tab:xi1} compares~\eqref{eq:xi1} against $\dd\xi/\dd M$ obtained by numerical
differentiation of the full nonlinear problem~\eqref{eq:pb}--\eqref{eq:robinp}. Agreement is within
$0.1$--$0.5\%$ on every branch tested, the residual being the $O(M^2)$ truncation of the finite
difference. The general form~\eqref{eq:xi1gen} was checked on the plane-wave branch ($\Om=1.3$) not only
for the parabolic profile ($\xi_1=-0.7471$ vs $-0.7470$, $0.02\%$) but also for a sinusoidal profile
$f=\sin\pi\eta$ ($-0.8232$ vs $-0.8230$, $0.03\%$) and a quartic $f=16\eta^2(1-\eta)^2$, confirming that
the closed form is not special to Poiseuille flow.

\begin{table}[htb]
\centering
\begin{tabular}{@{}ccccc@{}}
\toprule
$\Om$ & branch ($\xi_0$) & type & $\xi_1$ from~\eqref{eq:xi1} & $\xi_1$ numerical\\
\midrule
1.3 & 1.1165 & symmetric structural & $-0.1283$ & $-0.1290$\\
1.3 & 1.1837 & antisymmetric structural & $-0.1065$ & $-0.1068$\\
1.3 & 1.3267 & plane wave & $-0.7471$ & $-0.7468$\\
1.8 & 1.8069 & plane wave & $-1.1743$ & $-1.1750$\\
\bottomrule
\end{tabular}
\caption{Closed-form $O(M)$ correction~\eqref{eq:xi1} against numerically differentiated
$\dd\xi/\dd M|_{M=0}$, for $l=3$, $\eps=0.25$. The plane wave is an order of magnitude more
flow-sensitive than the structural branches, as predicted by the $I_f$ weighting.}
\label{tab:xi1}
\end{table}

\section{Second-order Mach correction and the convection/shear split}\label{sec:om2}

\subsection{Structure of the $O(M^2)$ problem}
Collecting terms of order $M^2$ in~\eqref{eq:pb} using~\eqref{eq:Aexp} and~\eqref{eq:shearexp} gives a
second forced problem of the same type,
\begin{equation}
L_0[p_2]=G(\eta;\xi_2),
\label{eq:M2}
\end{equation}
whose forcing $G$ involves $p_0$, $p_1$, $\xi_0$, $\xi_1$, $\xi_2$ and the profile. The same Fredholm
argument applies and yields $\xi_2$; the difference from first order is that $G$ contains $p_1$, which
is not available in closed form, so the second-order coefficient must in general be evaluated
numerically. Within $l^2A_2p_0$ the term $l^2\xi_0^2f^2p_0$ is the pure convection contribution
(present even for a plug flow, since it arises from squaring the Doppler shift $(\omega-k_xU)^2$),
while the terms carrying $p_1$ transmit the shear information.

\begin{remark}[Why $\xi_2$ is even in $M$]
Under $M\to-M$ the Doppler factor becomes $\sig\to\Om+M\xi f$, which is equivalent to reversing the
direction of propagation. The odd coefficient $\xi_1$ therefore changes sign (it is the
direction-dependent Doppler shift), while $\xi_2$ does not: it is a correction to the \emph{magnitude}
of the wave speed and is the same upstream and downstream.
\end{remark}

\subsection{Extraction}
Because $\xi_2$ is even in $M$, it is cleanly extracted from the full nonlinear solution by Richardson
extrapolation of a symmetric pair of Mach numbers,
\begin{equation}
\xi(+M)+\xi(-M)=2\xi_0+2M^2\xi_2+O(M^4)
\quad\Longrightarrow\quad
\xi_2=\frac{\xi(+M)+\xi(-M)-2\xi_0}{2M^2},
\label{eq:richardson}
\end{equation}
which cancels the $O(M)$ and $O(M^3)$ terms identically. The shear content of $\xi_2$ is isolated by a
\emph{shear-suppressed surrogate}: the computation is repeated with the shear term of~\eqref{eq:pb}
switched off (that is, with $f'$ set to zero in the operator while $f$ is retained in $\sig$), and the
difference $\xi_2-\xi_2^{\rm suppr}$ measures the part of the second-order coefficient carried by the
mean shear rather than by pure Doppler convection.

\subsection{Convection, shear, and what a plug-flow model actually costs}
Results at $\Om=1.3$, $l=3$, $\eps=0.25$ are given in Table~\ref{tab:xi2}. Two comparisons are made.
The shear-suppressed surrogate keeps the exact profile in the Doppler factor and removes only the shear
operator, and so measures the intrinsic weight of the shear term; a \emph{genuine} plug flow, by
contrast, replaces the parabolic profile altogether by a uniform stream, matched either on peak velocity
($M_u=M$) or on mass flux ($M_u=\tfrac23 M$), and is the simplification actually made in much
duct-acoustic practice.

\begin{table}[htb]
\centering
\resizebox{\textwidth}{!}{%
\begin{tabular}{@{}lcccccccc@{}}
\toprule
 & & & \multicolumn{2}{c}{shear diagnostic} & \multicolumn{3}{c}{genuine plug flow}\\
\cmidrule(lr){4-5}\cmidrule(lr){6-8}
branch ($\xi_0$) & $\xi_1$ & $\xi_2$ & suppr. & shear \% & peak & flux & peak err.\\
\midrule
$0.68503$ ($m{=}1$ anti.\ cut-on) & $-0.8517$ & $+0.2826$ & $+0.2985$ & $-5.6$ & $+1.1451$ & $+0.5075$ & $+305\%$\\
$1.11649$ (sym.\ structural) & $-0.1283$ & $-0.3251$ & $-0.4068$ & $-25.1$ & $-0.8831$ & $-0.3881$ & $+172\%$\\
$1.18369$ (anti.\ structural) & $-0.1065$ & $+0.2071$ & $+0.0883$ & $+57.4$ & $+0.4249$ & $+0.1904$ & $+105\%$\\
$1.32671$ (plane wave) & $-0.7471$ & $+0.8220$ & $+1.1317$ & $-37.7$ & $+2.1931$ & $+0.9703$ & $+167\%$\\
\bottomrule
\end{tabular}}
\caption{Second-order Mach coefficient $\xi_2$ and its decomposition at $\Om=1.3$, $l=3$, $\eps=0.25$.
``Suppr.'' is the shear-suppressed surrogate, which sets $f'=0$ in the operator while retaining $f$ in
$\sig$; ``shear\,\%'' is $(\xi_2-\xi_2^{\rm suppr})/\xi_2$. The two plug columns replace the profile by a
uniform flow matched on peak velocity ($M_u=M$) or on mass flux ($M_u=\tfrac23 M$); the last column is
the relative error of the matched-peak plug flow against $\xi_2$. The matched-flux errors are $+80$,
$+19$, $-8$ and $+18\%$ in row order.}
\label{tab:xi2}
\end{table}

Two conclusions follow. First, the shear contribution is not a small correction to the convective one:
it reaches $38\%$ of the total on the plane wave and $57\%$ on the antisymmetric structural branch, and
it changes sign between branches, so a shear-free model misrepresents the second-order behaviour and the
full profile must be retained. Second, the error incurred by a \emph{genuine} plug flow is larger still.
Matching on centreline velocity overstates the flow effect by a factor of two to four; matching on mass
flux, which is the better choice, recovers most but not all of the difference. A uniform-flow model is
therefore adequate at zeroth order, adequate at first order in the convection-dominated limit, but
quantitatively wrong at second order. Including $\xi_2$ materially extends the range of $M$ over which
the asymptotic wavenumber is quantitative, as shown in Fig.~\ref{fig:meanflow}(a) and quantified in
Section~\ref{sec:largeM}. As noted in Remark~\ref{rem:incompress}, these are the exact $O(M^2)$
coefficients of the incompressible-parabolic model; mean-flow compressibility enters at the same order
and is not included.

\section{The critical layer}\label{sec:cl}

Because the critical layer is unlikely to be familiar to every reader of a structural-acoustics paper,
this section is written to be self-contained: what the singularity is, where it appears, what it does
to the solution, how causality resolves the resulting ambiguity, and how viscosity regularises it.

\subsection{Origin of the singularity}
The Pridmore--Brown equation~\eqref{eq:pb} carries $\sig(\eta)$ in the denominator of its
first-derivative coefficient. If there exists $\eta_c\in[0,1]$ with
\begin{equation}
\sig(\eta_c)=\Om-M\xi f(\eta_c)=0,
\label{eq:clcond}
\end{equation}
the equation has a regular singular point there and smooth-coefficient ODE theory does not apply.
Dividing~\eqref{eq:clcond} by $\xi$ and restoring dimensions shows what the condition means physically:
\begin{equation}
\frac{\omega}{k_x}=U(y_c),
\end{equation}
that is, the axial phase speed of the wave equals the local speed of the mean flow at that height. A
fluid particle at $y_c$ travels with the wave and sees a steady, rather than oscillatory, disturbance.
Wave and flow are in resonance there, and irreversible exchange of energy between them becomes
possible. This is the acoustic counterpart of the critical layer of hydrodynamic stability
theory~\cite{case,maslowe,drazinreid}, and its role in duct acoustics has been developed
in~\cite{brambleydarau,king2022}.

\subsection{Onset condition and the branch-resolved map}
Since $f$ attains its maximum value $f_{\max}=1$ at the centreline $\eta=1/2$, the first critical layer
to appear as $M$ is raised does so at the centreline. Setting $\sig(1/2)=0$,
\begin{equation}
\boxed{\;M\ge M_{\rm crit}=\frac{\Om}{\xi}=c_{\rm ph},\;}
\label{eq:Mcrit}
\end{equation}
where $c_{\rm ph}=\Om/\xi$ is the non-dimensional axial phase speed. The condition is therefore simply
that the flow becomes as fast as the wave. Its consequences differ sharply from branch to branch:

\begin{itemize}[leftmargin=1.6em,itemsep=2pt]
\item \textbf{Structural wave}, $\xi_0=\Om^{1/2}$: $\;M_{\rm crit}=\Om/\Om^{1/2}=\sqrt\Om$. Below
coincidence ($\Om<1$) this is \emph{subsonic}, and a critical layer is reachable at a physically
accessible Mach number. Above coincidence the structural wave is supersonic and $M_{\rm crit}>1$.
\item \textbf{Plane wave}, $\xi_0=\Om$: $\;M_{\rm crit}=1$, sonic flow.
\item \textbf{Cut-on branches}, $\xi_0<\Om$: $\;M_{\rm crit}>1$, supersonic flow.
\end{itemize}

\noindent The conclusion is worth isolating: \emph{within this waveguide, a critical layer is
accessible at subsonic flow only for the structural wave below coincidence}. This is a genuinely
structural--acoustic statement, since coincidence is a property of the plates and not of the fluid;
it has no counterpart in a rigid or constant-impedance duct. Fig.~\ref{fig:meanflow}(b,c) shows the
threshold and the corresponding touching of $\sig(\eta)$ at the centreline.

\subsection{Local (Frobenius) structure}
Write $\zeta=\eta-\eta_c$ and expand $\sig$ about the critical point,
\begin{equation}
\sig(\eta)\simeq\sig_c'\,\zeta,\qquad
\sig_c'=-M\xi f'(\eta_c)=-4M\xi(1-2\eta_c).
\end{equation}
To leading order~\eqref{eq:pb} becomes
\begin{equation}
\hat p''+\frac{2}{\zeta}\hat p'+\bigl(-l^2\xi_0^2+O(\zeta)\bigr)\hat p=0 .
\label{eq:local}
\end{equation}
Substituting a Frobenius series $\hat p=\zeta^\rho\sum_{n\ge0}a_n\zeta^n$ with $a_0\ne0$, the lowest
power is $\zeta^{\rho-2}$ and the indicial equation is
\begin{equation}
\rho(\rho-1)+2\rho=\rho(\rho+1)=0
\quad\Longrightarrow\quad
\rho_1=0,\ \rho_2=-1 .
\label{eq:indicial}
\end{equation}
The exponents differ by an integer, so the second solution may carry a logarithm. Carrying the
recurrence further (\ref{app:frobenius}) shows that the would-be $\rho=-1$ behaviour is
absorbed into the regular solution and that the logarithmic term first enters at order $\zeta^3$:
\begin{equation}
\hat p(\zeta)=C_1\bigl[1+O(\zeta^2)\bigr]+C_2\bigl[\zeta^3\ln\zeta+\cdots\bigr].
\label{eq:frob}
\end{equation}
The effective exponent pair for the pressure is thus $\{0,3\}$, and the essential point is that
\emph{the pressure remains bounded at the critical layer}. This is a distinctive feature of the
acoustic problem: the pressure eigenfunction, which is the object whose eigenvalue is being computed,
does not blow up.

\subsection{The velocity singularity}
The cross-stream velocity, however, does. From~\eqref{eq:vfromp},
\begin{equation}
v'=\frac{\ii\,\partial_yp'}{\sigma\rho}
\;\propto\;\frac{\hat p'}{\sig}\;\sim\;\frac{\hat p'(\eta_c)}{\sig_c'\,\zeta}\;\sim\;\zeta^{-1},
\label{eq:vsing}
\end{equation}
because $\hat p'$ is $O(1)$ at $\eta_c$ while $\sig$ vanishes linearly. The singularity is an
integrable inverse power, and it is carried by the velocity field (and, in a Green's-function
treatment, by the continuous spectrum~\cite{case,brambleydarau}) rather than by the discrete pressure
mode. This asymmetry between pressure and velocity is what ultimately explains the wall-mode result of
Section~\ref{sec:viscous}.

\subsection{Causality and the branch of the logarithm}
The term $\zeta^3\ln\zeta$ in~\eqref{eq:frob} is ambiguous until a branch of the logarithm is chosen,
and different choices give physically different answers: absorption, amplification, or neither. The
ambiguity is resolved by causality, that is by asking which solution is reached by a disturbance
switched on in the remote past.

The standard device is the Landau (vanishing-viscosity) prescription: give the frequency a small
positive imaginary part, $\omega\to\omega+\ii\epsilon$ with $\epsilon\to0^+$, which represents a wave
that was absent long ago and has grown to its present amplitude. Under this shift the critical-layer
location moves off the real axis,
\begin{equation}
\eta_c(\epsilon)=\eta_c+\ii\frac{\epsilon}{\omega M\xi f_c'}+O(\epsilon^2),
\qquad\text{with }\operatorname{Im}\eta_c>0,
\end{equation}
that is, into the upper half $\eta$-plane. The real integration path therefore passes \emph{below} the
branch point, and continuing from $\eta<\eta_c$ to $\eta>\eta_c$ along the real axis winds clockwise
around it, giving
\begin{equation}
\boxed{\;\Delta\bigl(\ln\zeta\bigr)=-\ii\pi\;}
\label{eq:phasejump}
\end{equation}
across the layer. The sign is the one that corresponds to net \emph{absorption} of wave energy by the
mean flow: the phase jump endows $\hat p$ with a small imaginary part above $\eta_c$, which through the
solvability condition at the next order shifts $\operatorname{Im}\xi$ by an amount proportional to
$|p_0(\eta_c)|^2$.

That last proportionality is the key to the whole question of how much the critical layer matters. For
the structural--acoustic wall modes studied here, the eigenfunction is driven by the plates and has its
pressure concentrated near $\eta=0$ and $\eta=1$, whereas the critical layer sits at
$\eta_c\approx\tfrac12$. One might therefore expect the coupling between the discrete mode and the layer
to be weak. Section~\ref{sec:viscous} tests this directly and finds the absorption of the antisymmetric
wall mode indeed negligible, but for a sharper and more robust reason than smallness, namely the
parity of the mode (Section~\ref{sec:parity}).

\subsection{Two critical layers, and the degeneracy at onset}\label{sec:twolayers}
The centreline is special twice over. \emph{At} onset, $M=M_{\rm crit}$, the layer sits exactly at
$\eta=\tfrac12$, where $f'(\tfrac12)=0$. The linear coefficient of the local expansion,
$\sig_c'=-M\xi f'(\eta_c)$, then vanishes, and $\sig$ has a \emph{double} zero:
$\sig(\tfrac12)=\sig'(\tfrac12)=0$ with $\sig''(\tfrac12)=-M\xi f''=8M\xi\ne0$. The indicial analysis of
Section~\ref{sec:cl} must be redone about this degenerate point; a Frobenius analysis about this double zero
gives the effective exponent pair $\{0,5\}$ in place of $\{0,3\}$, and the dominant balance in the viscous
inner layer becomes $\sig_c''\delta^2Z^2\sim\nu/\delta^2$, so that the inner-layer thickness
\emph{thickens} from $\Rey^{-1/3}$ to
\begin{equation}
\delta_{\rm onset}\sim\Rey^{-1/4}
\label{eq:degenerate}
\end{equation}
precisely at $M=M_{\rm crit}$.

\emph{Above} onset the single layer splits into a pair. Since $f(\eta_c)=\Om/(M\xi)<1$ for $M>M_{\rm crit}$
and $f$ is symmetric about the centreline with a single interior maximum, there are then \emph{two}
critical layers,
\begin{equation}
\eta_c^{\pm}=\tfrac12\pm\Delta,\qquad
f\bigl(\eta_c^{\pm}\bigr)=\frac{\Om}{M\xi},
\label{eq:pair}
\end{equation}
placed symmetrically about $\eta=\tfrac12$. At the two locations the mean shear has opposite sign,
$f'(\eta_c^{-})=-f'(\eta_c^{+})>0$, so $\sig_c'$ has opposite signs and, by the causal prescription of
Section~\ref{sec:cl}, the two layers carry phase jumps of \emph{opposite} sign. The net response of a
discrete mode to the pair therefore depends on how the mode samples the two heights, that is, on the
parity of $p_0$ about the centreline. This is the geometric fact behind the absorption result of
Section~\ref{sec:viscous}: for a mode with a mid-plane pressure node the pair straddles a zero and
contributes nothing at onset, whereas for a mode with a mid-plane antinode it does not.

\subsection{Viscous regularisation}
In a real fluid the singularity is not present; viscosity smooths it over a thin inner layer. The
relevant balance in the inner variable $Z=(\eta-\eta_c)/\delta$ is between the vanishing inertia term
and viscous diffusion,
\begin{equation}
\underbrace{\sig_c'\delta Z}_{\text{inertia}}\ \sim\ \underbrace{\frac{\nu}{\delta^2}}_{\text{viscosity}}
\quad\Longrightarrow\quad
\delta^3\sim\frac{\nu}{\sig_c'}
\quad\Longrightarrow\quad
\boxed{\;\delta\sim \Rey^{-1/3},\;}\qquad \Rey=\frac{ca}{\nu},
\label{eq:delta}
\end{equation}
which is the classical Tollmien scaling for a viscous critical layer. Retaining only the leading
viscous and inertial terms in the inner region reduces the linearised Navier--Stokes system to
\begin{equation}
\frac{\dd^3\tilde v}{\dd Z^3}-\ii Z\frac{\dd\tilde v}{\dd Z}=0,
\label{eq:airy}
\end{equation}
whose solutions are Airy functions of complex argument, $\Ai(\ii^{2/3}Z)$ and $\Bi(\ii^{2/3}Z)$. The
causal solution is the one that decays away from $Z=0$, namely $\Ai$; as $\Rey\to\infty$ and
$\delta\to0$ it matches onto the outer Frobenius solution~\eqref{eq:frob} and reproduces the same
phase jump $-\ii\pi$. The inviscid causal prescription is therefore the correct vanishing-viscosity
limit, not merely a convenient convention. Fig.~\ref{fig:viscous} shows all three features: the
approach of the computed phase jump to $\pi$ in magnitude as the viscous smoothing is reduced, the
$\Rey^{-1/3}$ scaling of the inner-layer thickness, and the regularised $(\eta-\eta_c)^{-1}$ velocity
singularity.

\begin{figure}[htb]
\centering
\includegraphics[width=0.98\textwidth]{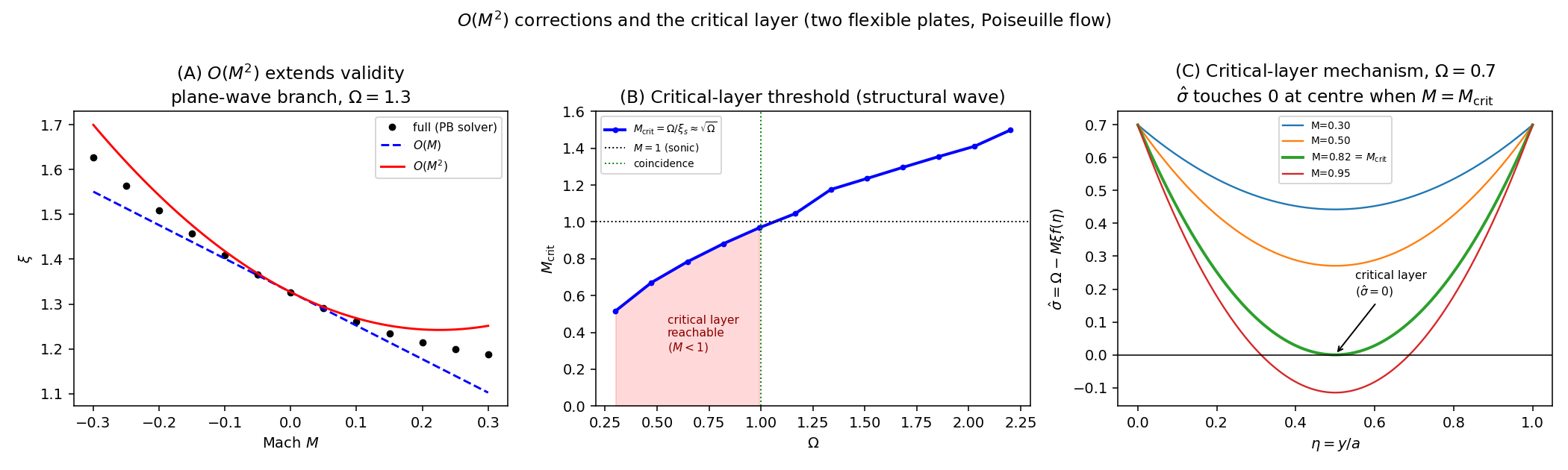}
\caption{Mean-flow effects. (a) Including the $O(M^2)$ term extends the range of Mach numbers over
which the asymptotic wavenumber tracks the exact solution. (b) The critical-layer threshold
$M_{\rm crit}=\Om/\xi_s\approx\sqrt\Om$ of~\eqref{eq:Mcrit}, reachable at subsonic $M$ only below
coincidence. (c) The Doppler-shifted frequency $\sig(\eta)$ touching zero at the centreline exactly
when $M=M_{\rm crit}$.}
\label{fig:meanflow}
\end{figure}

\begin{figure}[htb]
\centering
\includegraphics[width=0.98\textwidth]{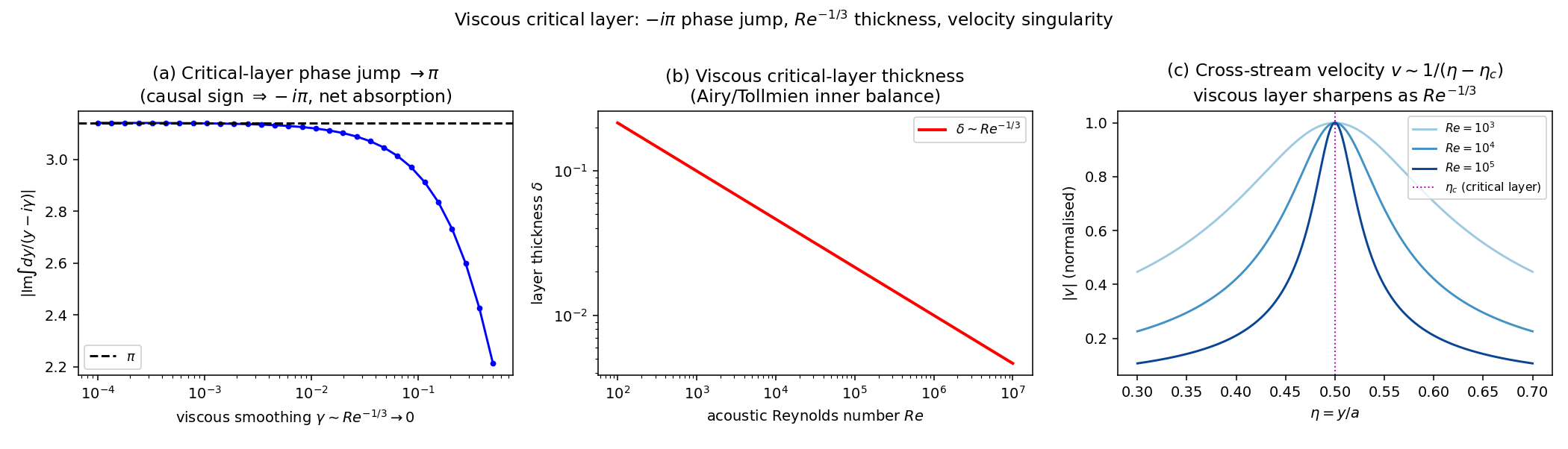}
\caption{Structure of the viscous critical layer. (a) The causal phase jump tends to $\pi$ in
magnitude as the viscous smoothing vanishes, i.e.\ $-\ii\pi$ as in~\eqref{eq:phasejump},
corresponding to net absorption. (b) Inner-layer thickness following $\delta\sim \Rey^{-1/3}$,
equation~\eqref{eq:delta}. (c) The cross-stream velocity singularity $v\sim(\eta-\eta_c)^{-1}$
of~\eqref{eq:vsing}, regularised over a layer that sharpens as $\Rey^{-1/3}$.}
\label{fig:viscous}
\end{figure}

\section{Viscous solution above the critical Mach number}\label{sec:viscous}

\subsection{Motivation}
Above $M_{\rm crit}$ the inviscid operator is singular and the Mach expansion of
Section~\ref{sec:lowmach} has no domain. The question of physical interest is not whether the expansion
survives (it does not) but whether the \emph{discrete wall modes} survive: does the critical layer
absorb them, as it can absorb the surface modes of an impedance-lined duct~\cite{brambleydarau}, or do
they pass through unaffected? The only way to answer this cleanly is to retain viscosity, which
regularises the operator, and then to examine the limit $\Rey\to\infty$.

\subsection{The viscous eigenvalue problem}
The linearised compressible Navier--Stokes equations, viscous in the momentum balance and inviscid in
the energy balance, are, in non-dimensional form with $\nabla\cdot\hat u=\ii\xi\hat u+l^{-1}\hat v_\eta$,
\begin{align}
\sig\hat u-Mf'\hat v&=-\ii\xi\hat p
+\frac{1}{\Rey}\Bigl[\hat u_{\eta\eta}-\xi^2\hat u+\frac{\ii\xi}{3}\nabla\cdot\hat u\Bigr],
\label{eq:lnsu}\\
\sig\hat v&=-\frac{1}{l}\hat p_\eta
+\frac{1}{\Rey}\Bigl[\hat v_{\eta\eta}-\xi^2\hat v+\frac{1}{3l}\partial_\eta(\nabla\cdot\hat u)\Bigr],
\label{eq:lnsv}\\
\sig\hat p+\frac{\gamma-1}{\gamma}\Bigl(\ii\xi\hat u+\frac1l\hat v_\eta\Bigr)
&=\frac{\gamma-1}{\gamma\Rey\mathit{Pr}}\bigl(\hat p_{\eta\eta}-\xi^2\hat p\bigr).
\label{eq:lnsp}
\end{align}
For the conclusions drawn below the energy equation is treated as inviscid, that is the right-hand side
of~\eqref{eq:lnsp} is dropped. This is justified because the viscous corrections to the pressure field
are $O(\Rey^{-1})$, whereas the critical-layer structure is governed by the momentum balance
at $O(\Rey^{-1/3})$ and the wall damping at $O(\Rey^{-1/2})$; the neglected terms are subdominant to
both.

\subsection{Viscous boundary conditions}
With viscosity retained the velocity is no longer slaved to $\hat p_\eta$, and no-slip must be imposed
separately, $\hat u(0)=\hat u(1)=0$. The pressure is then recovered from the continuity
relation~\eqref{eq:lnsp} rather than from the $y$-momentum equation:
\begin{equation}
\hat p=\frac{\xi\hat u}{\sig}+\frac{\ii\hat v_\eta}{l\sig}
\quad\overset{\hat u=0,\ \sig=\Om}{\Longrightarrow}\quad
\hat p(0)=\frac{\ii\hat v_\eta(0)}{l\Om}.
\end{equation}
Substituting this into the non-dimensional admittance relation
$\hat v(0)=-\ii\eps\hat p(0)/(l\Om S_1)$ converts the plate condition into a Robin condition on the
normal velocity,
\begin{equation}
\boxed{\;\hat v(0)-\beta_1\hat v_\eta(0)=0,\qquad
\hat v(1)+\beta_2\hat v_\eta(1)=0,\qquad
\beta_j=\frac{\eps}{l^2\Om^2S_j}.\;}
\label{eq:robinv}
\end{equation}
Together with no-slip these supply the four conditions the fourth-order viscous system requires. As
$\Rey\to\infty$ the viscous $y$-momentum balance returns to its inviscid form
$\hat v=\hat p_\eta/(l\Om\sig)$, and substituting this into~\eqref{eq:robinv} recovers the inviscid
pressure conditions~\eqref{eq:robinp} exactly; the reduction is verified symbolically and set out in
\ref{app:bc}. The two formulations are therefore consistent, and any difference between the
inviscid and viscous eigenvalues is a genuine viscous effect and not an artefact of the wall condition.

\subsection{Discretisation}
The system~\eqref{eq:lnsu}--\eqref{eq:lnsp} with its four boundary conditions is discretised by
Chebyshev collocation on the Gauss--Lobatto points $\eta_k=\bigl(1-\cos(k\pi/N)\bigr)/2$,
$k=0,\dots,N$, with the standard differentiation matrices $D^{(1)}$ and $D^{(2)}$. Stacking the
unknowns $(\hat u_k,\hat v_k,\hat p_k)$ into a vector $\mathbf q$ gives
\begin{equation}
\mathsf A(\xi)\,\mathbf q=\mathbf 0,
\end{equation}
in which $\mathsf A$ depends nonlinearly on $\xi$ both through $\sig=\Om-M\xi f$ and through
$S_j(\xi)$ in the boundary conditions. The dispersion relation is $\det\mathsf A(\xi)=0$, solved by
Newton iteration continued from the quiescent eigenvalue $\xi_0$. Resolution of the critical layer
requires $N$ large enough that several collocation points fall within the inner layer of thickness
$\Rey^{-1/3}$; $N=200$--$400$ was used, with convergence checked by doubling $N$.

\subsection{Results}
Four findings emerge, shown in Fig.~\ref{fig:viscglobal}. Unless stated otherwise they are for the
antisymmetric structural branch continued from $\xi_0=0.858$ at $\Om=0.7$, $l=3$, $\eps=0.25$, whose
critical layer is reached at $M_{\rm crit}\approx0.82$; the symmetric branch at the same frequency is
treated separately in Section~\ref{sec:parity}.

\paragraph{(1) The eigenvalue is regular through $M_{\rm crit}$.} $\operatorname{Re}\xi$ varies
smoothly as $M$ passes through $M_{\rm crit}$, with no jump, cusp or turning point
(Fig.~\ref{fig:viscglobal}a, $\Rey=3000$). Viscosity smooths the singularity over the inner layer of
thickness $\Rey^{-1/3}$, and the discrete mode does not notice the transition.

\paragraph{(2) The decay rate scales as $\Rey^{-1/2}$.} The attenuation $\operatorname{Im}\xi$ obeys a
clean $\Rey^{-1/2}$ law both below ($M=0.70$) and above ($M=0.92$) the threshold
(Fig.~\ref{fig:viscglobal}b). This exponent identifies the mechanism unambiguously as the wall Stokes
layers, whose thickness is $\delta_S\sim(\nu/\omega)^{1/2}\sim \Rey^{-1/2}$ in units of $a$: the
viscous stress at the wall is $\mu\partial_yu'\sim\mu u'_{\rm wall}/\delta_S$, and balancing the work
done against the wave energy, which scales as $|p_{\rm wall}|^2$, gives $\operatorname{Im}\xi\sim
\Rey^{-1/2}$. Critically, this damping \emph{vanishes} in the inviscid limit.

\paragraph{(3) The antisymmetric wall mode is not absorbed.} Fitting
\begin{equation}
\operatorname{Im}\xi=A\,\Rey^{-1/2}+B
\label{eq:AB}
\end{equation}
separates the wall-Stokes contribution, which disappears as $\Rey\to\infty$, from the
$\Rey$-\emph{independent} part $B$. The latter is the candidate critical-layer absorption, since the
critical layer is an inviscid feature whose effect must survive the inviscid limit. For the
antisymmetric branch the fit returns $|B|\approx1\times10^{-5}$ both below ($M=0.70$) and above
($M=0.92$) onset (Fig.~\ref{fig:viscglobal}c), at the resolution floor of the fit. Three caveats bound
this number. It is not zero in an absolute sense: at the largest Reynolds number computed, $\Rey=10^4$,
the total attenuation is $|\operatorname{Im}\xi|\approx2\times10^{-4}$, so $|B|$ is a few per cent of the
measured damping, and it \emph{changes sign} between the two Mach numbers, evidence that it is at the
resolution limit rather than resolved. A two-term fit may, moreover, absorb an intermediate $\Rey^{-1/3}$
scaling into $A$ and $B$. And the limit $\Rey\to\infty$ is not physically attainable: the hydrodynamic
Reynolds number $\Rey_h=M\Rey$ reaches the plane-Poiseuille Tollmien--Schlichting threshold
$\Rey_h\approx1.15\times10^4$ at $\Rey\approx1.25\times10^4$ for $M=0.92$, so~\eqref{eq:AB} is a trend
over a finite window, not a limit. Within these limits, the interior critical layer does not measurably
absorb the antisymmetric wall mode.

\paragraph{(4) Convective stability.} The inviscid antisymmetric eigenvalues remain real throughout the
range studied, and the viscous decay rate satisfies $\operatorname{Im}\xi>0$, so these waves decay
downstream: in the classification of Benjamin~\cite{benjamin} and Landahl~\cite{landahl} they behave as
damped class-B disturbances, the damping supplied by the wall Stokes layers and vanishing inviscidly. No
flow-induced flutter of the antisymmetric wall modes is found up to the transonic limit of the model.
The symmetric branch behaves differently above onset, as Sections~\ref{sec:parity} and~\ref{sec:instab} show.

\begin{figure}[htb]
\centering
\includegraphics[width=0.98\textwidth]{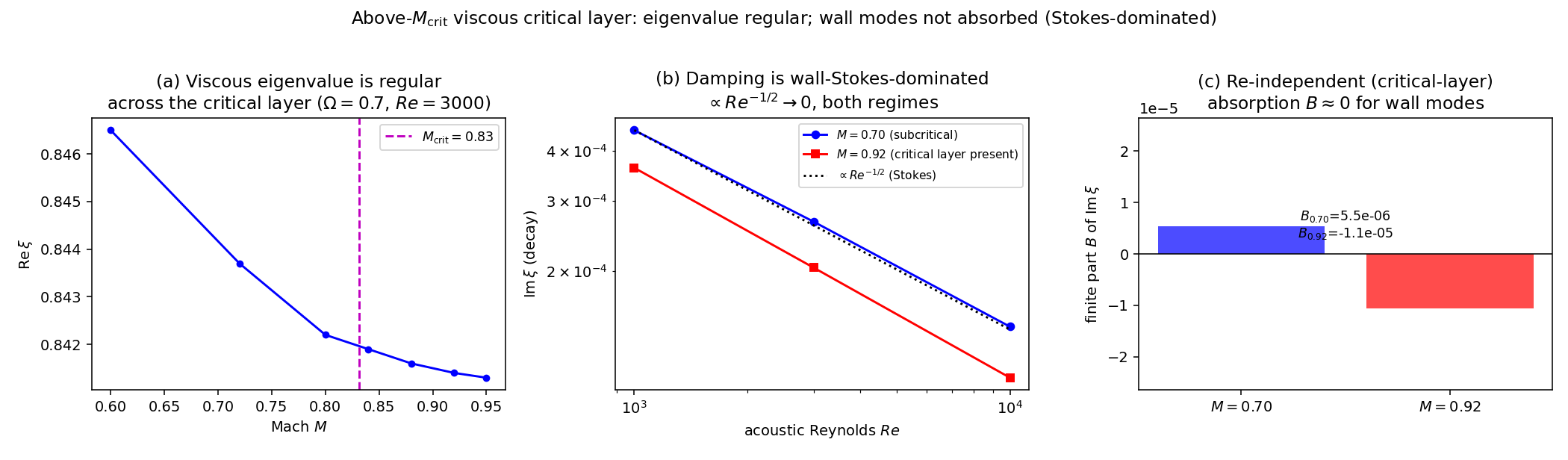}
\caption{Viscous solution through and above $M_{\rm crit}$ for the antisymmetric structural branch
($\Om=0.7$, $l=3$, $\eps=0.25$). (a) The eigenvalue is regular through the critical layer ($\Rey=3000$).
(b) The decay rate follows $\operatorname{Im}\xi\propto\Rey^{-1/2}$ both below ($M=0.70$) and above
($M=0.92$) $M_{\rm crit}$, identifying it as wall-Stokes-layer damping that vanishes in the inviscid
limit. (c) The $\Rey$-independent residual $B$ is at the resolution floor, $|B|\approx10^{-5}$, and
changes sign between the two Mach numbers: the antisymmetric wall mode is not measurably absorbed.}
\label{fig:viscglobal}
\end{figure}

\subsection{Why the absorption vanishes: parity, not smallness}\label{sec:parity}
The absorption implied by the phase jump is proportional to $|p_0(\eta_c)|^2$, and it is natural to
attribute the null of item~(3) to smallness: wall-driven modes concentrate their pressure at the plates,
while the layer sits near $\eta=\tfrac12$, so the overlap is small. At the parameter point computed that
argument is weaker than it looks, and the true reason is both simpler and stronger.

\paragraph{The mode is not strongly wall-concentrated.} The symmetric structural mode at this frequency
has $p_0(\tfrac12)/p_0(0)=0.746$: three quarters of the wall amplitude survives at the centreline. For
that branch the overlap integral is emphatically not small, and a smallness argument would predict a
resolvable absorption.

\paragraph{The mode computed has an exact node.} The branch continued in Fig.~\ref{fig:viscglobal} is
the antisymmetric structural wave, for which the mid-plane is an exact pressure node,
$p_0(\tfrac12)\equiv0$, by the factorisation~\eqref{eq:symanti}. At onset, where $\eta_c=\tfrac12$, the
overlap integral therefore vanishes \emph{identically by symmetry}, not approximately by localisation.
Above onset the layer splits into the symmetric pair~\eqref{eq:pair}, at which $p_0$ no longer vanishes;
but for the antisymmetric mode $p_0(\eta_c^{-})=-p_0(\eta_c^{+})$ while the two phase jumps are opposite
in sign (Section~\ref{sec:twolayers}), so the contributions cancel. Either way the net absorption is
zero for reasons of parity.

\paragraph{The decisive test, and its outcome.} These observations turn the numerical null into a
symmetry statement and identify the two computations that decide whether it generalises. Both were
carried out with a \emph{parity-restricted} viscous solver: for identical plates the discretised
operator commutes exactly with the mirror reflection $\eta\to1-\eta$ (the commutation residual is
$\sim10^{-16}$), so each symmetry family can be solved in isolation. This is essential here, because
above onset the symmetric and antisymmetric structural eigenvalues become nearly degenerate
($\xi\approx0.841$ for both at $M=0.92$) and a naive tracker hops between them; the parity restriction
keeps them apart.
\begin{enumerate}[label=(\alph*),leftmargin=2.2em,itemsep=3pt]
\item \emph{Symmetric branch.} Repeating the traverse on the symmetric structural branch
($\xi_0=0.88169$, $M_{\rm crit}=0.794$), which has no parity node, returns a \emph{resolvable}
$\Rey$-independent residual $|B|\approx6.4\times10^{-4}$ at $M=0.92$, about sixty times the
antisymmetric floor. It is $N$-converged (identical to five significant figures for $N=110$, $150$ and
$190$) and unchanged whether or not an intermediate $\Rey^{-1/3}$ term is included in the fit, and
$|\operatorname{Im}\xi|$ grows towards $|B|$ as $\Rey$ increases, the signature of a genuine inviscid
residual rather than a numerical artefact. Below onset ($M=0.70$) the same branch sits at the floor.
The null is therefore \emph{not} a geometric property of the waveguide; it is specific to the
antisymmetric parity.
\item \emph{Broken mirror symmetry.} Making the plates unequal, $r_D\ne1$, destroys the exact
factorisation~\eqref{eq:factor}: the antisymmetric-derived mode acquires a nonzero centreline pressure
and its residual lifts off the floor, from $|B|\approx1\times10^{-5}$ at $r_D=1$ to $\approx3\times10^{-4}$
once $|r_D-1|\gtrsim0.005$, tracking the induced $p_0(\tfrac12)$. The lift is steep and then saturates,
reflecting the sharp avoided crossing of the two nearly-degenerate branches; it is not the slow linear
growth a well-separated pair would show.
\end{enumerate}
Figure~\ref{fig:parity} collects the three panels: the parity of the two eigenfunctions, the
extrapolation of $\operatorname{Im}\xi$ to its inviscid residual on each branch, and the growth of $|B|$
as the mirror symmetry is broken.

The conclusion is sharper than ``wall modes are not absorbed.'' For the antisymmetric structural mode
the interior critical layer produces no measurable absorption, and this is protected by the mode's
parity rather than by the weakness of the singularity; the symmetric mode, which shares the frequency
but not the node, enjoys no such protection. The resolvable symmetric-branch residual carries the sign
of amplification rather than of decay, so the critical layer feeds energy into the symmetric wall mode.
Section~\ref{sec:instab} establishes that this is a genuine convective instability of that mode.

\begin{figure}[htb]
\centering
\includegraphics[width=0.98\textwidth]{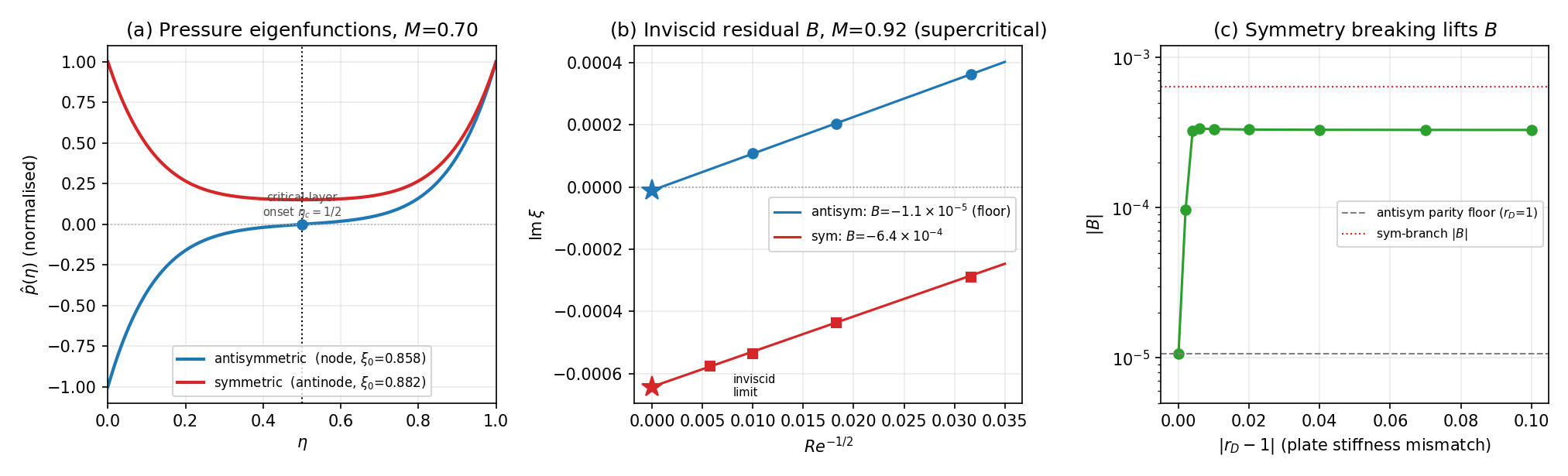}
\caption{The parity test of the absorption null ($\Om=0.7$, $l=3$, $\eps=0.25$), computed with the
parity-restricted viscous solver. (a) Pressure eigenfunctions at $M=0.70$: the antisymmetric branch has
an exact node at the critical-layer onset location $\eta=\tfrac12$, the symmetric branch an antinode.
(b) The $\Rey$-independent residual $B$ at $M=0.92$: the antisymmetric branch extrapolates to the
resolution floor ($|B|\approx10^{-5}$), while the symmetric branch extrapolates to a resolvable
$B\approx-6.4\times10^{-4}$. (c) Breaking the mirror symmetry ($r_D\ne1$) lifts $|B|$ off the
antisymmetric floor towards the symmetric-branch value and saturates, confirming that the null is a
parity effect rather than a geometric one.}
\label{fig:parity}
\end{figure}

\subsection{A convective instability of the symmetric wall mode}\label{sec:instab}
Section~\ref{sec:parity} left one question open: the resolvable symmetric-branch residual $B$ carries the
sign of amplification, not of decay. This section resolves it. Four computations, all with the
parity-restricted solver, establish that the symmetric structural wall mode is convectively unstable
above $M_{\rm crit}$, that the instability survives the inviscid limit, that its onset is set by the
critical layer, and that it is distinct from the Tollmien--Schlichting and Kelvin--Helmholtz mechanisms.

\paragraph{The sign, settled two ways.} The wall Stokes term is dissipative, so its sign in the solver
is, by definition, decay; in the fit $\operatorname{Im}\xi=A\Rey^{-1/2}+B$ this term is $A$, positive on
both branches. On the symmetric branch the Reynolds-independent residual $B=-6.5\times10^{-4}$ carries
the opposite sign and dominates at every $\Rey$, so the net contribution is anti-dissipative. A
convention-independent temporal calculation confirms it: fixing the real axial wavenumber and solving
for the complex frequency $\Omega$, the viscously damped antisymmetric branch returns
$\operatorname{Im}\Omega<0$ with a magnitude that falls as $\Rey^{-1/2}$ and vanishes as $\Rey\to\infty$,
whereas the symmetric branch returns $\operatorname{Im}\Omega>0$, the opposite sign, with a magnitude
that grows with $\Rey$ toward a finite limit $\omega_i^\infty\approx1.06\times10^{-3}$
(Fig.~\ref{fig:instab}b). Growth that strengthens toward the inviscid limit identifies an inviscid
critical-layer instability, not a viscous or numerical artefact.

\paragraph{Onset is set by the critical layer.} Tracking the temporal growth rate of the symmetric
branch in $M$ at fixed frequency (Fig.~\ref{fig:instab}a) shows it rising monotonically from negative,
where the mode is viscously damped, and crossing zero into growth as $M$ increases past the
critical-layer threshold. Repeating this at several frequencies gives the neutral curve $M_n(\Om)$ of
Fig.~\ref{fig:stab}a, which lies just above the critical-layer onset $M_{\rm crit}(\Om)=\Om/\xi_0$ and
follows it, from $M_n=0.759$ at $\Om=0.55$ to $M_n=0.904$ at $\Om=0.75$. Below $M_{\rm crit}$ there is no
critical layer and the mode is damped; above it the layer feeds energy to the mode. The small offset
$M_n-M_{\rm crit}$ is the residual wall-Stokes damping the forcing must overcome, and it narrows as
$\Rey$ increases.

\paragraph{The instability is convective.} At the unstable point the group velocity is
$c_g=\dd\Om/\dd\xi=1.65$, real and of order unity, and it does not vanish anywhere on the unstable
branch, so the branch supports no zero-group-velocity saddle. The leading and trailing edges of the wave
packet travel at $v_+=1.75$ and $v_-=1.55$, both positive, so the entire packet is swept downstream while
it amplifies. In the Briggs--Bers sense the instability is therefore convective, not absolute: a
disturbance grows as it is carried along the channel but decays at any fixed station, as for
Tollmien--Schlichting waves in an unbounded shear flow.

\paragraph{A compliant-wall mode, not Tollmien--Schlichting or Kelvin--Helmholtz.} Three facts separate
it from the two classical channel-flow instabilities. First, the growth rate is proportional to the wall
compliance, $\operatorname{Im}\Omega\approx2.9\times10^{-3}\,\eps$ through the origin
(Fig.~\ref{fig:stab}b), so it vanishes as the plates are made rigid, whereas a Tollmien--Schlichting wave
exists for rigid walls and is independent of $\eps$. Second, the instability is present at hydrodynamic
Reynolds numbers far below the Tollmien--Schlichting threshold: it is resolved already at $\Rey=10^3$,
where $\Rey_h=M\Rey=920$, while plane Poiseuille flow is linearly stable to Tollmien--Schlichting waves
below $\Rey_h\approx1.15\times10^4$. Third, the Kelvin--Helmholtz mechanism is absent by construction: it
requires a slip velocity at the wall, and the no-slip Poiseuille profile annihilates the convective
(Ingard--Myers) term exactly (Remark~\ref{rem:IM}), so there is no vortex sheet to roll up. The unstable
mode is a discrete compressible structural--acoustic wall mode of the compliant channel, destabilised by
the interior critical layer.

\paragraph{Consequence.} The statement of Sections~\ref{sec:viscous} and~\ref{sec:largeM} that the system
is convectively stable holds for the antisymmetric branch only. The symmetric structural wall mode is
convectively unstable above $M_{\rm crit}$, and this critical-layer instability, absent for a rigid wall
and for a slip profile, is the principal physical result of the paper.

\begin{figure}[htb]
\centering
\includegraphics[width=0.98\textwidth]{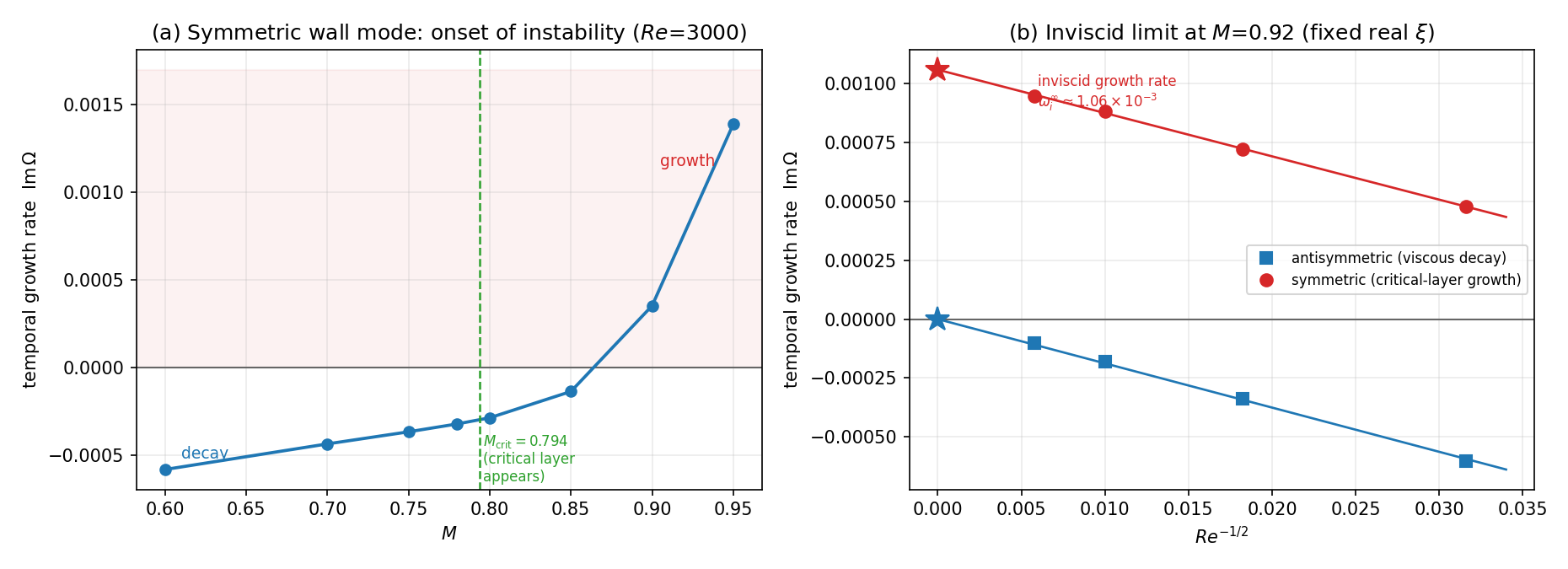}
\caption{The convective critical-layer instability of the symmetric wall mode ($\Om=0.7$, $l=3$,
$\eps=0.25$), parity-restricted solver. (a) Temporal growth rate $\operatorname{Im}\Omega$ against Mach
number at $\Rey=3000$: negative (viscously damped) below the critical-layer threshold $M_{\rm crit}=0.794$
and positive (growth) above it. (b) At $M=0.92$, the temporal growth rate against $\Rey^{-1/2}$ for the
two branches at fixed real wavenumber: the antisymmetric branch is viscously damped and extrapolates to
zero, while the symmetric branch grows to a finite inviscid limit $\omega_i^\infty\approx1.06\times10^{-3}$.}
\label{fig:instab}
\end{figure}

\begin{figure}[htb]
\centering
\includegraphics[width=0.98\textwidth]{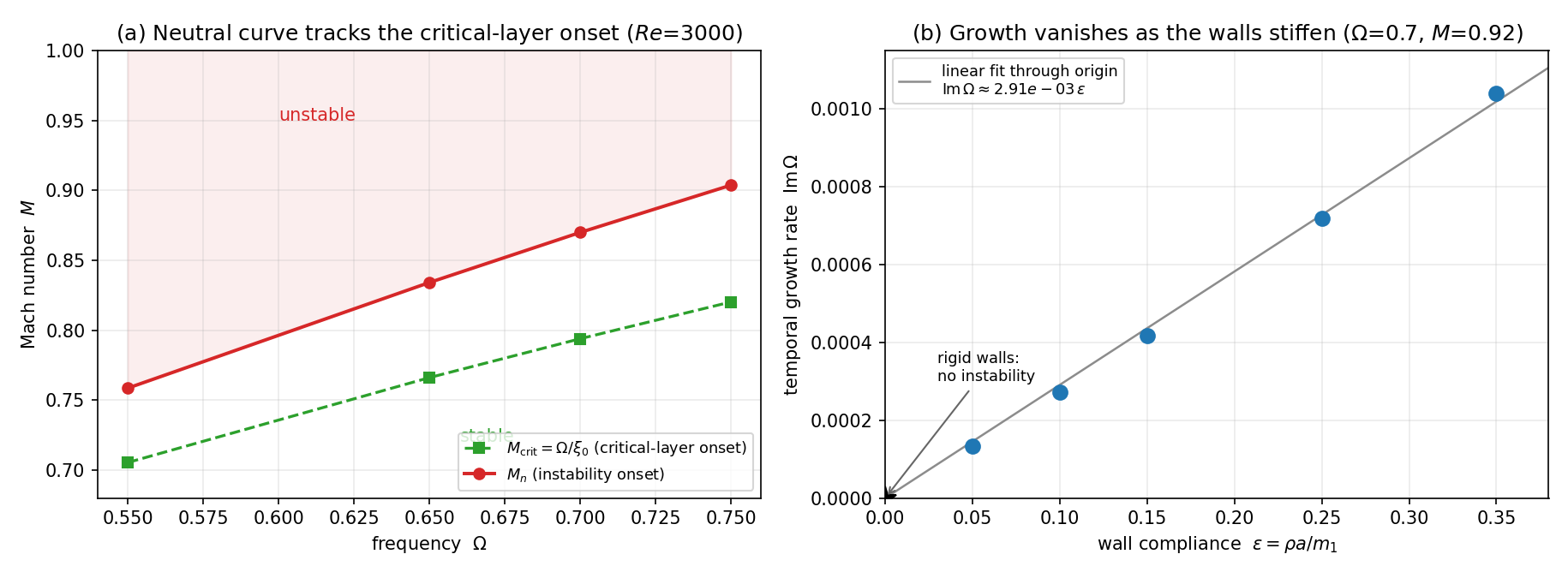}
\caption{(a) Neutral curve $M_n(\Om)$ of the symmetric-mode instability (red) and the critical-layer onset
$M_{\rm crit}(\Om)=\Om/\xi_0$ (green), $\Rey=3000$: the instability sets in just above the critical layer
and follows it. (b) The growth rate at $\Om=0.7$, $M=0.92$ is proportional to the wall compliance $\eps$
and passes through the origin, so the instability vanishes for rigid walls, which distinguishes it from
the Tollmien--Schlichting mechanism.}
\label{fig:stab}
\end{figure}

\section{Numerical methods and validation}\label{sec:val}

The analytical results are supported by three numerical schemes that share no code and rest on
different formulations of the problem.

\paragraph{Method 1: Runge--Kutta shooting.} Equation~\eqref{eq:pb} is integrated from $\eta=0$ to
$\eta=1$ with a fourth-order Runge--Kutta scheme, starting from the normalisation $\hat p(0)=1$ and the
boundary condition $\hat p'(0)=-\eps\hat p(0)/S_1$. The residual at the far wall,
$R(\xi)=\hat p'(1)-(\eps/S_2)\hat p(1)$, is driven to zero by Newton iteration on $\xi$. Step-size
halving is used to confirm convergence, and near $M_{\rm crit}$ the integration path is indented into
the complex $\eta$-plane according to the Landau prescription of Section~\ref{sec:cl}.

\paragraph{Method 2: Chebyshev spectral collocation.} The pressure is expanded in $N$ Chebyshev
polynomials on $[0,1]$ and~\eqref{eq:pb} with~\eqref{eq:robinp} is enforced at the collocation points,
converting the problem into a matrix eigenvalue problem solved by the QZ algorithm. The physical
eigenvalue is identified by continuation from the quiescent solution.

\paragraph{Method 3: primitive-variable formulation.} The system is written directly in the three
fields $(\hat u,\hat v,\hat p)$ of the linearised Euler equations~\eqref{eq:pbu}--\eqref{eq:pbc},
without eliminating $\hat u$ and $\hat v$ a priori, and the resulting first-order system is solved by
Chebyshev collocation. Reducing the same system symbolically to a single equation in $\hat p$
reproduces the Pridmore--Brown operator~\eqref{eq:pb} exactly; it was this reduction that fixed the sign
of the shear term.

\paragraph{Agreement.} The three methods agree on every branch to approximately $6\times10^{-9}$ for
Mach numbers up to $M=0.2$, a level consistent with double-precision round-off in the shooting
integration. At $M=0$ all three reduce to the closed-form quiescent relation~\eqref{eq:disp}, whose
algebra was itself verified with a symbolic-algebra package. As a direct test of the asymptotics, the
closed-form $\xi_1$ of~\eqref{eq:xi1} matches numerically differentiated $\dd\xi/\dd M$ to
$0.1$--$0.5\%$ (Table~\ref{tab:xi1}). Finally, the symmetric branch of the quiescent problem reproduces
the single fluid-loaded-plate result of~\cite{sarkarsonti} under the substitution
$(a,\eps)\to(a/2,\eps/2)$ to four decimal places at every frequency tested
(Fig.~\ref{fig:validation}).

\begin{figure}[htb]
\centering
\includegraphics[width=0.9\textwidth]{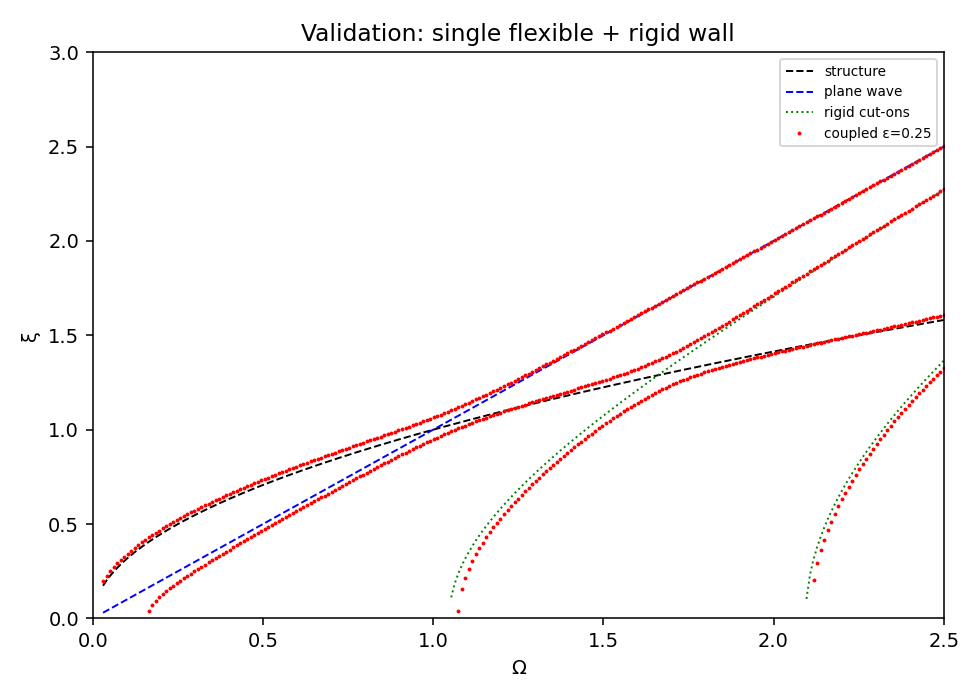}
\caption{Validation against the single-flexible-plate problem. The present solver, run in the
rigid-cover limit $\alpha_2\to\infty$, reproduces the coupled dispersion curves of~\cite{sarkarsonti}:
the coupled branches sit just off the in-vacuo flexural, plane-wave and rigid cut-on parents, with gaps
at coincidence and at each cut-on crossing.}
\label{fig:validation}
\end{figure}

\section{Finite Mach numbers: full dispersion, validity and stability}\label{sec:largeM}

\subsection{Full nonlinear dispersion}
The corrections of Sections~\ref{sec:lowmach}--\ref{sec:om2} are asymptotic in $M$, and to see what
happens at larger Mach numbers the dispersion must be read directly off~\eqref{eq:pb} without
expansion. Fig.~\ref{fig:Mmap} does this for $M=0$, $0.3$ and $0.5$. An increasing downstream flow
pushes the whole diagram towards lower $\xi$, in agreement with Proposition~\ref{prop:sign}; the shift
is largest for the plane-wave-like branch and smallest for the structural branches, in agreement with
the $I_f$ weighting of Section~\ref{sec:lowmach}; and the interleaved cut-on ladder structure survives
intact at every $M$ studied.

One limitation should be stated plainly. A strictly parabolic profile is the incompressible Poiseuille
solution and is an idealisation at finite Mach number, so the model is quantitative only for $M\ll1$.
The finite-$M$ results are offered as an extension of the model, useful for identifying trends and for
locating the critical layer, and not as a prediction for a genuinely compressible mean flow, whose
profile would itself deform.

\subsection{Validity of the asymptotic expansions}
How far the closed forms can be trusted depends strongly on the branch
(Fig.~\ref{fig:valid}a,b). On the flow-sensitive plane-wave branch both the $O(M)$ and $O(M^2)$
expansions pass a five-percent relative error already at $M\approx0.25$--$0.3$; being asymptotic rather
than convergent, the series ultimately diverges, and beyond $M\sim0.3$ adding terms makes matters
worse. On a flow-insensitive structural branch the same expansions stay within about three percent out
to considerably larger $M$, because the structural wave speed is set mainly by plate stiffness rather
than by the flow. Cut-on branches are intermediate, with the effective expansion parameter being
$M\xi_0/\Om$ rather than $M$ itself, that is, the ratio of the flow speed to the axial phase speed,
which is exactly $M/M_{\rm crit}$. The practical rule is therefore: use the closed forms for
$M\lesssim0.3\,M_{\rm crit}$ and the full solver beyond.

\subsection{Convective stability}
The convective stability of the discrete wall modes was checked with the viscous solver of
Section~\ref{sec:viscous}. On the antisymmetric branch the inviscid eigenvalues remain real and the
viscous decay rate satisfies $\operatorname{Im}\xi>0$ throughout the range studied
(Fig.~\ref{fig:valid}c), so the mean flow does not destabilise those modes; the damping present
originates in the wall Stokes layers and vanishes as the viscosity is removed. The symmetric structural
branch is different: above $M_{\rm crit}$ it is convectively unstable, as Section~\ref{sec:instab}
establishes. Both statements concern the discrete spectrum only and say nothing about the continuous
spectrum, which is not treated here.

\begin{figure}[htb]
\centering
\includegraphics[width=0.72\textwidth]{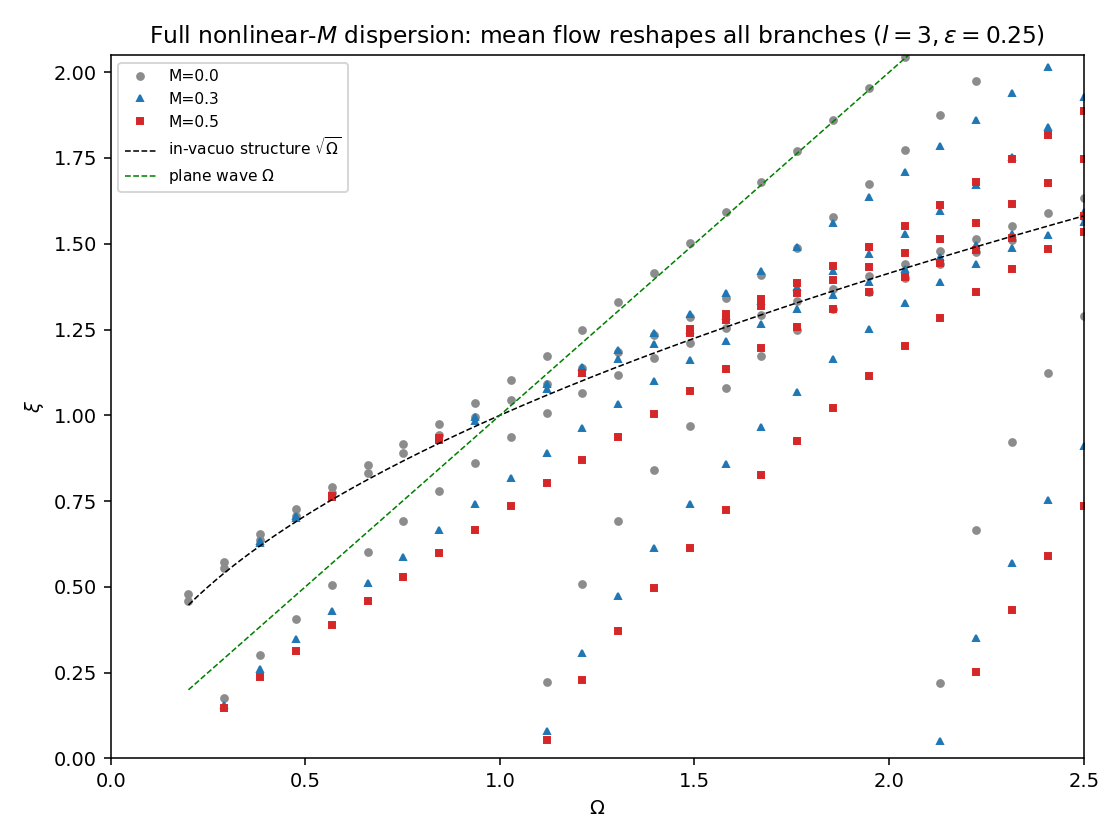}
\caption{Full nonlinear-in-$M$ dispersion ($l=3$, $\eps=0.25$). The mean flow shifts every coupled
branch to lower $\xi$ with increasing downstream Mach number, most strongly near the plane wave and
least on the structural branches; the interleaved cut-on ladder structure is preserved.}
\label{fig:Mmap}
\end{figure}

\begin{figure}[htb]
\centering
\includegraphics[width=0.99\textwidth]{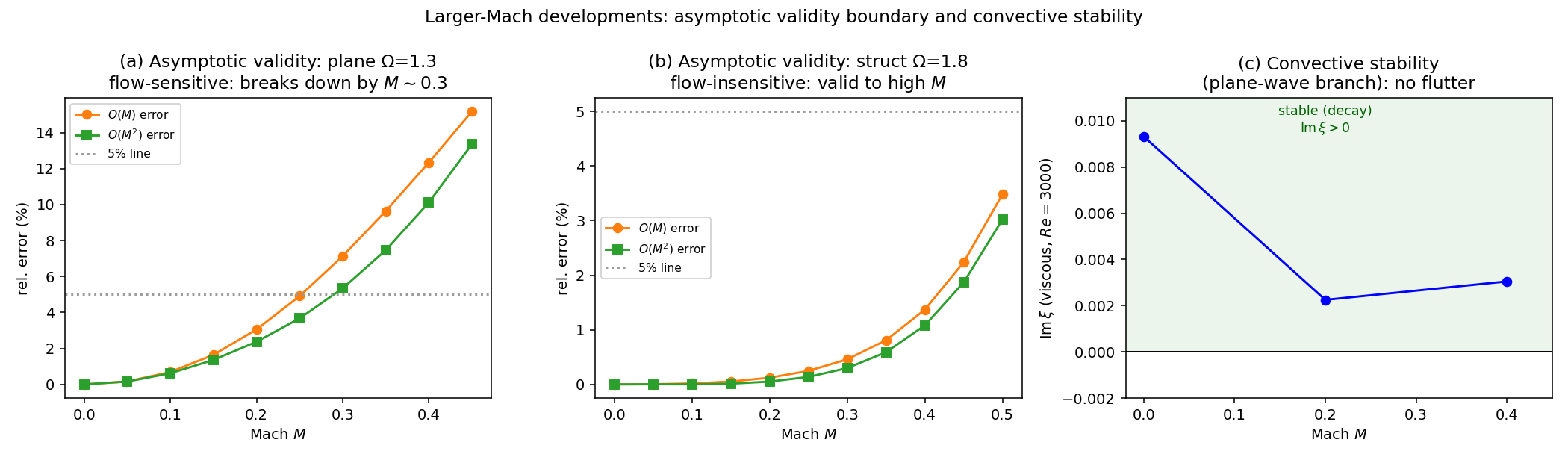}
\caption{(a,b) Branch-dependent validity of the $O(M)$ and $O(M^2)$ asymptotics: the flow-sensitive
plane-wave branch breaks down by $M\sim0.3$, while a flow-insensitive structural branch remains
accurate to considerably larger $M$. (c) The antisymmetric viscous decay rate satisfies
$\operatorname{Im}\xi>0$ throughout; the symmetric branch is treated in Section~\ref{sec:instab}.}
\label{fig:valid}
\end{figure}

\begin{figure}[htb]
\centering
\includegraphics[width=0.9\textwidth]{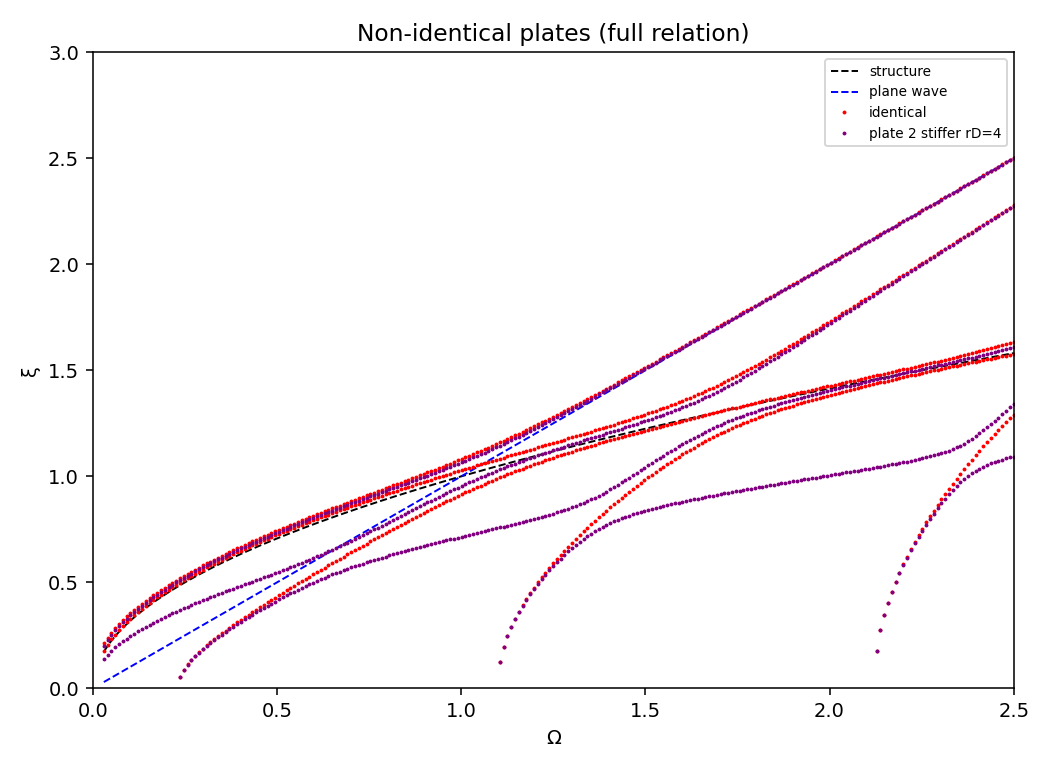}
\caption{Non-identical plates. Comparison of the identical case $r_D=1$ with a stiffer second plate
$r_D=4$, computed from the full relation~\eqref{eq:disptan}. The symmetry decoupling is lost and the
two structural waves veer rather than cross, with the minimum separation given
by~\eqref{eq:veer}.}
\label{fig:nonident}
\end{figure}

\section{Discussion}\label{sec:disc}

\subsection{What the mean flow does}
The mean flow enters the problem in three separable ways, and the closed form~\eqref{eq:xi1} makes the
separation explicit. Doppler convection, weighted by $I_f=\int fp_0^2$, dominates for modes that
sample the fast core; boundary shear, weighted by the squared wall pressures, dominates for modes
concentrated at the plates; and bulk shear, weighted by $I_0$, acts uniformly as a distributed added
mass. Because the three carry different weights on different branches, the flow sensitivity of the
spectrum is strongly non-uniform: at $\Om=1.3$, $l=3$, $\eps=0.25$ the plane wave has
$|\xi_1|\approx0.75$ while the structural branches have $|\xi_1|\approx0.11$--$0.13$, a factor of six.
The practical consequence for a designer of a compliant duct is that flow affects the acoustic branches
of the spectrum far more than the structural ones, and that the two can therefore be tuned
semi-independently.

At second order the picture changes qualitatively rather than merely quantitatively, because the shear
contribution is no longer a correction to the convective one. Table~\ref{tab:xi2} shows the shear part
reaching $38\%$ of the total on the plane wave and $57\%$ on the antisymmetric structural branch and
changing sign between branches, while a genuine plug flow (a common simplification in duct-acoustic
practice) errs by a factor of two to four when matched on centreline velocity. Any uniform-flow model
is therefore wrong at $O(M^2)$, by an amount that differs in sign from branch to branch.

\subsection{What the critical layer does, and does not, do}
The critical-layer threshold~\eqref{eq:Mcrit} divides the parameter space into a regular subsonic
regime, in which the low-Mach asymptotics apply, and a singular regime beyond it. The branch-resolved
form of the threshold is the structural--acoustic content of the result: because $M_{\rm crit}=\Om/\xi$
and the structural wave has $\xi=\Om^{1/2}$, the threshold for that wave is $\sqrt\Om$, which is below
unity precisely when the frequency is below coincidence. Coincidence is a property of the plates, so
the accessibility of a critical layer in this waveguide is controlled by the structure, not by the
fluid, a feature with no counterpart in a rigid or lined duct.

Within the singular regime the layer has the classical structure: bounded pressure with a
$\zeta^3\ln\zeta$ term, a $\zeta^{-1}$ velocity divergence, a causal phase jump of $-\ii\pi$, and a
viscous inner layer of thickness $\Rey^{-1/3}$ that thickens to $\Rey^{-1/4}$ at the degenerate onset
point. What is specific to this problem, however, is not weak coupling but \emph{parity}. Absorption is
proportional to $|p_0(\eta_c)|^2$, and the antisymmetric structural mode (the one branch for which
the critical layer is subsonically accessible) carries an exact pressure node at the very location
where the layer first appears; above onset the two symmetric layers contribute equal and opposite phase
jumps. The null is therefore a symmetry statement, confirmed by a parity-restricted computation that
also shows the symmetric mode, with its centreline antinode, to be unprotected
(Section~\ref{sec:parity}).

The contrast with the impedance-wall problem of Brambley, Darau and Rienstra~\cite{brambleydarau} is
instructive. There the critical layer sits inside a thin boundary layer adjacent to the lining, where
the surface modes have their largest amplitude, and the absorption is significant. Here the shear
occupies the whole channel, the layer is at the centreline, and, crucially, the channel possesses
a mirror symmetry that an impedance-lined duct does not; both differences point the same way for the
antisymmetric mode. The broader lesson is that the importance of a critical layer to a discrete mode is
governed less by the strength of the singularity than by the symmetry of the mode relative to the layer.

\subsection{Where the singular response goes}
If the discrete modes are not absorbed, the energy exchange implied by the $-\ii\pi$ phase jump must be
carried elsewhere. It is carried by the cross-stream velocity field, which diverges as
$(\eta-\eta_c)^{-1}$, and in a Green's-function treatment by the continuous spectrum: the non-modal
contribution obtained by integrating around the branch cut~\cite{case,brambleydarau,king2022}. King
\emph{et al.}~\cite{king2022} found this contribution to dominate the downstream field for a quadratic
boundary-layer profile, which is the profile relevant here, so it is entirely possible that the
critical-layer response of the present waveguide is large in a forced problem while being negligible
in the modal problem treated above. The two statements are compatible, and distinguishing them
carefully is the natural next step.

\subsection{Limitations}
Four limitations should be stated. First, the mean flow is incompressible and strictly parabolic, so
the finite-$M$ results are an extension of the model rather than a compressible prediction. Second,
only the discrete spectrum is computed; the continuous spectrum, which carries the genuinely singular
critical-layer response, is not. Third, the plates are modelled by Kirchhoff theory with no in-plane
tension, no rotary inertia and no shear deformation, which restricts the analysis to wavelengths long
compared with the plate thickness. Fourth, the critical layer is treated linearly; at finite amplitude
the $(\eta-\eta_c)^{-1}$ velocity singularity is regularised by nonlinear wave--wave interaction rather
than by viscosity, with vortex roll-up and wave breaking, and the nonlinear layer is outside the scope
of the present linear analysis.

\section{Conclusions}\label{sec:concl}

This paper has analysed the coupled dispersion of a two-dimensional structural--acoustic waveguide
formed by a compressible fluid column between two thin elastic plates, carrying a fully developed
Poiseuille mean flow. The principal results are as follows.

\begin{enumerate}[label=\arabic*.,leftmargin=2em,itemsep=3pt]
\item The governing equations have been set out in full. With a sheared mean flow the pressure
perturbation obeys the Pridmore--Brown equation~\eqref{eq:pb}, derived here from the linearised Euler
system for this geometry. Because a Poiseuille profile satisfies no-slip, the convective
(Ingard--Myers) term in the compliant-wall boundary condition vanishes \emph{identically}, so that the
plate-admittance conditions retain their quiescent Robin form~\eqref{eq:robinp} at every Mach number.
The mean flow enters only through the bulk operator. This exact property of the configuration removes
the well-posedness difficulties associated with the Ingard--Myers condition and is what makes the
closed-form analysis possible.

\item The leading operator of the Mach expansion is self-adjoint under the plate-admittance conditions
(Theorem~\ref{thm:sa}), so the Fredholm alternative applies with $p_0$ as its own adjoint null
function. The resulting solvability condition yields the fully explicit first-order
correction~\eqref{eq:xi1gen}, valid on every branch, for arbitrary fluid loading, for non-identical
plates and for an arbitrary smooth no-slip shear profile, and requiring only three elementary
quadratures of the quiescent mode shape (two for the parabolic profile). The correction is
negative on every branch (Proposition~\ref{prop:sign}), so a downstream flow lowers the axial
wavenumber and the dispersion relation loses its $k_x\to-k_x$ symmetry. It decomposes cleanly into
boundary-shear, bulk-shear and Doppler-convection contributions, of which the last dominates and
explains why the plane wave is roughly six times more flow-sensitive than the structural branches.

\item The second-order coefficient $\xi_2$ is even in $M$ and has been split into convection and shear
parts. The shear part reaches $38\%$ of the total on the plane-wave branch and $57\%$ on the
antisymmetric structural branch and changes sign between branches; a genuine plug flow errs by a factor
of two to four, so a uniform-flow model misrepresents the second-order behaviour and the full profile
must be retained.

\item A critical layer appears once $M\ge M_{\rm crit}=\Om/\xi$. Because $\xi=\Om^{1/2}$ on the
structural wave, $M_{\rm crit}=\sqrt\Om$ there, which is subsonic precisely below coincidence; every
other branch requires sonic or supersonic flow. The accessibility of a critical layer in this waveguide
is therefore controlled by a property of the structure. The layer has bounded pressure with a
$\zeta^3\ln\zeta$ term, a $\zeta^{-1}$ cross-stream velocity divergence, a causal phase jump of
$-\ii\pi$, and a viscous inner layer of thickness $\Rey^{-1/3}$ governed by an Airy balance. The onset
point is degenerate ($\sig$ has a double zero, the exponents become $\{0,5\}$ and the inner layer
thickens to $\Rey^{-1/4}$), and above onset there are two critical layers, symmetric about the
centreline, carrying phase jumps of opposite sign.

\item Solving the viscous linearised Navier--Stokes problem shows that the discrete wall modes pass
through and above $M_{\rm crit}$ with a regular eigenvalue whose attenuation scales as $\Rey^{-1/2}$,
identifying it as wall Stokes-layer damping that vanishes inviscidly. For the antisymmetric structural
mode the $\Rey$-independent critical-layer absorption is at the resolution floor ($|B|\approx10^{-5}$),
and this null is shown to be a consequence of \emph{parity}: the mode's exact mid-plane pressure node
coincides with the onset location, and the two supercritical layers carry opposite phase jumps. A
parity-restricted computation (necessary because the two structural branches become nearly
degenerate above onset and a naive tracker conflates them) shows that the symmetric mode, which has a
centreline antinode, instead acquires a resolvable residual about sixty times larger, and that breaking
the mirror symmetry with unequal plates lifts the antisymmetric residual off the floor. The absorption
null is therefore specific to the antisymmetric family and symmetry-protected, not a general geometric
smallness, which accounts for the contrast with the impedance-wall surface modes
of~\cite{brambleydarau}. The antisymmetric wall modes are convectively stable up to the transonic limit;
the symmetric structural mode, by contrast, is convectively unstable above $M_{\rm crit}$. This is a
critical-layer instability of a compliant-wall structural--acoustic mode: its growth survives the
inviscid limit, its neutral curve $M_n(\Om)$ follows the critical-layer onset $M_{\rm crit}(\Om)$, and
its growth rate is proportional to the wall compliance, so it vanishes for rigid walls and is distinct
from the Tollmien--Schlichting and Kelvin--Helmholtz instabilities.

\item The quiescent baseline (the coupled dispersion relation~\eqref{eq:disp}, its exact
symmetric/antisymmetric factorisation, and the small- and large-$\eps$ catalogue) recovers for the
two-plate geometry results consistent with the established asymptotic-\linebreak[1]tracking framework, and is
presented here in full as setup, notation and $M=0$ baseline rather than as a claim to novelty. Its
incremental content is the antisymmetric family, invisible in the one-flexible-plate geometry, and the
non-identical-plate veering~\eqref{eq:veer}.

\item Three independent numerical schemes agree to $\sim6\times10^{-9}$, and the closed-form $\xi_1$
matches numerical differentiation to $0.1$--$0.5\%$.
\end{enumerate}

Two questions are left open. The first is the continuous spectrum, which carries the genuinely singular
part of the critical-layer response and which, on the evidence of~\cite{king2022} for a quadratic
profile, may dominate the forced problem even though it is irrelevant to the modal one. The second is the
nonlinear critical layer, in which the velocity singularity is regularised by wave--wave interaction
rather than by viscosity. Both lie outside the linear, modal framework used here, and the nonlinear
saturation of the convective instability of Section~\ref{sec:instab} is the natural next step.

\section*{Declaration of competing interest}
\addcontentsline{toc}{section}{Declaration of competing interest}
The author declares no known competing financial interests or personal relationships that could have
appeared to influence the work reported in this paper.

\section*{Data availability}
\addcontentsline{toc}{section}{Data availability}
The scripts implementing the closed-form catalogue, the three numerical solvers (including the
parity-restricted viscous eigensolver of Section~\ref{sec:viscous}) and the figure generation are openly
available at \url{https://github.com/i4mGr0ot/Acoustics} and archived at
\href{https://doi.org/10.5281/zenodo.22288008}{doi:10.5281/zenodo.22288008}.

\appendix

\section{Derivation of the quiescent asymptotic catalogue}\label{app:catalogue}

This appendix supplies the algebra behind Section~\ref{sec:catalogue}. Throughout,
$S(\xi)=\xi^4/\Om^2-1$ and $\tau=l\sqrt{\Om^2-\xi^2}$, and the two identical-plate relations are
\eqref{eq:nd}.

\subsection*{A.1 Structural branch, small $\eps$}
Set $\xi=\Om^{1/2}+a_1\eps+O(\eps^2)$. Since $S(\Om^{1/2})=0$ and
$S'(\xi)=4\xi^3/\Om^2$, one has $S=4\Om^{-1/2}a_1\eps+O(\eps^2)$. The transverse wavenumber evaluated
on the parent is $\tau_s=l\sqrt{\Om^2-\Om}$, and $\tau=\tau_s+O(\eps)$. Substituting into the
symmetric relation $S\tau\tan(\tau/2)+\eps=0$ and balancing at $O(\eps)$,
\begin{equation}
4\Om^{-1/2}a_1\,\tau_s\tan\frac{\tau_s}{2}+1=0
\quad\Longrightarrow\quad
a_1=-\frac{\Om^{1/2}}{4\tau_s\tan(\tau_s/2)}
=-\frac{1}{4l\sqrt{\Om-1}\,\tan(\tau_s/2)},
\end{equation}
using $\tau_s=l\Om^{1/2}\sqrt{\Om-1}$. This is~\eqref{eq:struct}$_1$. The antisymmetric relation
$S\tau\cot(\tau/2)-\eps=0$ gives, by the same balance,
$a_1=+\tan(\tau_s/2)/[4l\sqrt{\Om-1}]$, which is~\eqref{eq:struct}$_2$. As $\Om\to1$, $\tau_s\to0$ and
$\tan(\tau_s/2)\to\tau_s/2$, so the symmetric $a_1$ diverges (signalling the breakdown treated in
A.4) while the antisymmetric $a_1\to\tau_s/(8l\sqrt{\Om-1})=\tfrac18$, which is finite.

\subsection*{A.2 Plane wave, small $\eps$}
Set $\xi=\Om+b_1\eps$. Then $\tau^2=l^2(\Om^2-\xi^2)=-2l^2\Om b_1\eps+O(\eps^2)$, so $\tau=O(\sqrt\eps)$
and $\tau\tan(\tau/2)\to\tau^2/2$. With $S(\Om)=\Om^2-1$ the symmetric relation gives
\begin{equation}
(\Om^2-1)\cdot\bigl(-l^2\Om b_1\eps\bigr)+\eps=0
\quad\Longrightarrow\quad
b_1=\frac{1}{\Om l^2(\Om^2-1)},
\end{equation}
which is~\eqref{eq:planewave}. Note that the antisymmetric relation has no such root, because
$\tau\cot(\tau/2)\to2$ as $\tau\to0$ and the $O(\eps^0)$ balance cannot be satisfied.

\subsection*{A.3 Cut-on branches, small $\eps$}
Set $\tau=m\pi+t$ with $t=O(\eps)$, so that $\tan(\tau/2)\approx t/2$ for $m$ even and
$\cot(\tau/2)\approx-t/2$ for $m$ odd. Since $\tau^2=l^2(\Om^2-\xi^2)$,
$t=-l^2\xi_0(\xi-\xi_0)/(m\pi)$ with $\xi_0=\sqrt{\Om^2-(m\pi/l)^2}$. Substituting into the
appropriate relation and balancing at $O(\eps)$ gives, in both cases,
\begin{equation}
\xi-\xi_0=\frac{2\eps}{S_0l^2\xi_0},\qquad S_0=\frac{\xi_0^4}{\Om^2}-1,
\end{equation}
which is~\eqref{eq:cuton}.

\subsection*{A.4 Local $\sqrt\eps$ expansions}
Near a crossing the two contributions to the dispersion relation vanish simultaneously and the regular
balance fails. At symmetric coincidence, $\Om=1$, both $S(\xi)$ and $\tau$ are small; writing
$\xi=\Om+\sqrt\eps X$ gives $\tau^2\approx-2l^2\Om\sqrt\eps X$ and
$S\approx4(\xi-\Om^{1/2})\approx4\sqrt\eps X$ at $\Om=1$, so the relation
$S\tau\tan(\tau/2)+\eps=0$ becomes at $O(\eps)$
\begin{equation}
4X\cdot\bigl(-l^2 X\bigr)+1=0
\quad\Longrightarrow\quad
X=\pm\frac{1}{2l},
\end{equation}
which is~\eqref{eq:gapcoinc}. At a cut-on crossing the structural wave $\xi=\sqrt{\Om_c}$ meets the
$m$-th cut-on $\tau=m\pi$; writing $\xi=\sqrt{\Om_c}+\sqrt\eps X$ and expanding both $S$ and
$\tan(\tau/2)$ (or $\cot$) linearly gives the quadratic $2l^2X^2=1$, hence
$X=\pm\sqrt{2\eps}/(2l)$ after restoring the scaling, which is~\eqref{eq:gapcuton}. In both cases the
two roots $\pm X$ are the two sides of the gap, and continuity in $\Om$ shows that the branches
exchange identity across it.

\subsection*{A.5 Heavy loading}
Setting $\eps'=1/\eps$ and dividing~\eqref{eq:nd} through by $\eps$, the symmetric relation becomes
$\eps'S\tau\tan(\tau/2)+1=0$, whose $\eps'=0$ limit requires $\tan(\tau/2)\to\infty$, i.e.
$\tau=(2m{+}1)\pi$, the pressure-release ladder. Expanding $\tau=(2m{+}1)\pi+t$ with $t=O(\eps')$ and
balancing gives~\eqref{eq:heavy}. The antisymmetric relation gives $\cot(\tau/2)\to\infty$, i.e.
$\tau=2m\pi$, including $m=0$. This is the parity exchange recorded in Table~\ref{tab:migration}.

\section{Quadratures for standard mode shapes}\label{app:integrals}

The closed form~\eqref{eq:xi1} requires only $p_0(0)$, $p_0(1)$, $I_0=\int_0^1p_0^2$ and
$I_f=\int_0^1fp_0^2$ with $f=4\eta(1-\eta)$; the $O(M^2)$ estimate additionally uses
$I_{f^2}=\int_0^1f^2p_0^2$. All are elementary.

\paragraph{Plane wave, $p_0\equiv1$.}
\begin{equation}
p_0(0)=p_0(1)=1,\qquad I_0=1,\qquad
I_f=\int_0^14\eta(1-\eta)\,\dd\eta=4\Bigl(\frac12-\frac13\Bigr)=\frac23,\qquad
I_{f^2}=\frac{8}{15}.
\end{equation}

\paragraph{Cosine mode, $p_0=\cos(\tau_0\eta)$.}
\begin{align}
I_0&=\int_0^1\cos^2(\tau_0\eta)\,\dd\eta=\frac12+\frac{\sin2\tau_0}{4\tau_0},\\
I_f&=\int_0^14\eta(1-\eta)\cos^2(\tau_0\eta)\,\dd\eta
=\frac23-\frac{2}{\tau_0^2}+\frac{\sin2\tau_0}{2\tau_0^3}+\frac{\cos2\tau_0-1}{2\tau_0^2}.
\end{align}

\paragraph{General mode, $p_0=A\cos(\tau_0\eta)+B\sin(\tau_0\eta)$.}
\begin{equation}
I_0=\frac{A^2+B^2}{2}+\frac{(A^2-B^2)\sin2\tau_0+2AB(1-\cos2\tau_0)}{4\tau_0},
\end{equation}
\begin{equation}
p_0(0)^2=A^2,\qquad p_0(1)^2=\bigl(A\cos\tau_0+B\sin\tau_0\bigr)^2 .
\end{equation}
The flow-weighted norm $I_f$ decomposes into $\int_0^1f\cos^2$, $\int_0^1f\sin^2$ and
$\int_0^1f\cos\sin$, each of which reduces to elementary functions by two integrations by parts using
$f''=-8$.

\section{Frobenius analysis at the critical layer}\label{app:frobenius}

With $\zeta=\eta-\eta_c$ and $\sig\simeq\sig_c'\zeta$, the Pridmore--Brown equation~\eqref{eq:pb}
becomes
\begin{equation}
\hat p''+\frac{2}{\zeta}\hat p'+\bigl(-l^2\xi_0^2+O(\zeta)\bigr)\hat p=0 .
\label{eq:appfrob}
\end{equation}
Substituting $\hat p=\zeta^\rho\sum_{n\ge0}a_n\zeta^n$ with $a_0\ne0$,
\begin{equation}
\hat p''=\sum_na_n(\rho+n)(\rho+n-1)\zeta^{\rho+n-2},
\qquad
\frac{2}{\zeta}\hat p'=\sum_n2a_n(\rho+n)\zeta^{\rho+n-2},
\end{equation}
and the coefficient of the lowest power $\zeta^{\rho-2}$ gives the indicial equation
\begin{equation}
a_0\bigl[\rho(\rho-1)+2\rho\bigr]=a_0\,\rho(\rho+1)=0,
\qquad \rho_1=0,\quad\rho_2=-1 .
\end{equation}
The exponents differ by the integer $1$, so the second solution may contain a logarithm. Two features
of the recurrence determine the outcome. First, for $\rho_1=0$ the constant term $-l^2\xi_0^2$ feeds
into the recurrence only from $n=2$, so the regular solution is $1+O(\zeta^2)$. Second, for
$\rho_2=-1$ the coefficient $a_1$ is forced to diverge, which is the standard indication that the
second solution is non-analytic and that a logarithmic term is required. Constructing it in the form
$\hat p^{(2)}=\hat p^{(1)}\ln\zeta+\sum_nb_n\zeta^{n-1}$ and matching order by order shows that the
logarithm first survives at order $\zeta^3$ in the pressure, so that
\begin{equation}
\hat p=C_1\bigl[1+O(\zeta^2)\bigr]+C_2\bigl[\zeta^3\ln\zeta+O(\zeta^4)\bigr],
\end{equation}
which is~\eqref{eq:frob}: the effective exponent pair for the pressure is $\{0,3\}$ and the pressure is
bounded. The singular behaviour of the physical fields therefore enters not through $\hat p$ but
through the explicit factor $1/\sig\sim1/\zeta$ appearing in~\eqref{eq:vfromp}
and~\eqref{eq:ufromp}, which is why $v'\sim\zeta^{-1}$ while $\hat p=O(1)$.

\section{Equivalence of the inviscid and viscous plate conditions}\label{app:bc}

The inviscid formulation applies the Robin condition~\eqref{eq:robinp} to the pressure; the viscous
formulation applies~\eqref{eq:robinv} to the normal velocity. This appendix shows they agree in the
inviscid limit.

Starting from the common admittance relation~\eqref{eq:adm}, non-dimensionalised at $\eta=0$,
\begin{equation}
\hat v(0)=-\frac{\ii\eps}{l\Om S_1}\hat p(0).
\label{eq:appadm}
\end{equation}
In the viscous formulation, no-slip $\hat u=0$ together with the continuity
relation~\eqref{eq:lnsp} at $\sig=\Om$ gives $\hat p(0)=\ii\hat v_\eta(0)/(l\Om)$. Substituting
into~\eqref{eq:appadm},
\begin{equation}
\hat v(0)=-\frac{\ii\eps}{l\Om S_1}\cdot\frac{\ii\hat v_\eta(0)}{l\Om}
=\frac{\eps}{l^2\Om^2S_1}\hat v_\eta(0)
\quad\Longrightarrow\quad
\hat v(0)-\beta_1\hat v_\eta(0)=0,\qquad \beta_1=\frac{\eps}{l^2\Om^2S_1},
\end{equation}
which is~\eqref{eq:robinv}$_1$; the condition at $\eta=1$ follows identically with the sign of the
outward normal reversed, giving $\hat v(1)+\beta_2\hat v_\eta(1)=0$.

Conversely, as $\Rey\to\infty$ the viscous $y$-momentum balance~\eqref{eq:lnsv} returns to its
inviscid form $\sig\hat v=-\hat p_\eta/l$, so that at the wall, where $\sig=\Om$,
$\hat v(0)=-\hat p_\eta(0)/(l\Om)$. Substituting this into~\eqref{eq:appadm} gives
\begin{equation}
-\frac{\hat p_\eta(0)}{l\Om}=-\frac{\ii\eps}{l\Om S_1}\hat p(0)\cdot\frac{1}{\ii}
\quad\Longrightarrow\quad
\hat p_\eta(0)+\frac{\eps}{S_1}\hat p(0)=0,
\end{equation}
recovering~\eqref{eq:robinp}$_1$. The two formulations are therefore the same condition expressed in
different variables, and the difference between the inviscid and viscous eigenvalues reported in
Section~\ref{sec:viscous} is a genuine viscous effect rather than an artefact of the wall model. Both
reductions, and the coefficients $\beta_j$, were verified with a symbolic-algebra package. Together
with no-slip $\hat u(0)=\hat u(1)=0$, the conditions~\eqref{eq:robinv} supply the four boundary
conditions required by the fourth-order viscous system.

\section{Lagrange identity}\label{app:lagrange}

For $u,v\in C^2[0,1]$ and any constant $\tau$,
\begin{equation}
\int_0^1v\bigl(u''+\tau^2u\bigr)\dd\eta=\bigl[vu'-uv'\bigr]_0^1+\int_0^1u\bigl(v''+\tau^2v\bigr)\dd\eta .
\end{equation}
\begin{proof}
$\int_0^1vu''\,\dd\eta=[vu']_0^1-\int_0^1v'u'\,\dd\eta
=[vu']_0^1-[v'u]_0^1+\int_0^1v''u\,\dd\eta=[vu'-v'u]_0^1+\int_0^1uv''\,\dd\eta$.
Adding $\tau^2\int_0^1vu\,\dd\eta$ to both sides gives the result.
\end{proof}

\end{document}